\documentclass[a4paper,abstracton]{scrartcl}
\usepackage{amsfonts}
\usepackage{amssymb}
\usepackage{amsmath}
\usepackage{amsthm}
 \usepackage{algorithm}

\usepackage{bm}

\usepackage{graphicx}
\usepackage{color}

\usepackage{hyperref}

\usepackage{caption}
\usepackage{subcaption}

\usepackage{a4wide}

\newtheorem{theorem}{Theorem}

\newtheorem{lemma}[theorem]{Lemma}

\newtheorem{corollary}[theorem]{Corollary}

\usepackage{algpseudocode}
\algdef{SE}[DOWHILE]{Do}{doWhile}{\algorithmicdo}[1]{\algorithmicwhile\ #1}%

\usepackage{authblk}

\allowdisplaybreaks

\usepackage{natbib}
\begin{document}

%%%%%%%%%%%%%%%%%%%%%%%%%%%%%%%%%%%%%%%%%%%%%%%%%%%%%%%%%%%%%%%%%
%%%%%%  Macros
%%%%%%%%%%%%%%%%%%%%%%%%%%%%%%%%%%%%%%%%%%%%%%%%%%%%%%%%%%%%%%%%%
\newcommand{\X}{{\mathcal{X}}}
\newcommand{\cU}{{\mathcal{U}}}
\newcommand{\cI}{{\mathcal{I}}}
\newcommand{\cC}{{\mathcal{C}}}
 
%Highlighted changes in response to reviewer comments
\newcommand{\change}{\color{blue}}   
%%%%%%%%%%%%%%%%%%%%%%%%%%%%%%%%%%%%%%%%%%%%%%%%%%%%%%%%%%%%%%%%%
\title{Investigating the Recoverable Robust Single Machine Scheduling Problem Under Interval Uncertainty}

\author[1]{Matthew Bold\footnote{Corresponding author, email: m.bold1@lancaster.ac.uk}}
\author[2]{Marc Goerigk}

\affil[1]{STOR-i Centre for Doctoral Training, Lancaster University, Lancaster, UK}	
\affil[2]{Network and Data Science Management, University of Siegen, Siegen, Germany}

\date{}

\maketitle

\begin{abstract}
	We investigate the recoverable robust single machine scheduling problem under interval uncertainty. In this setting, jobs have first-stage processing times $\bm{p}$ and second-stage processing times $\bm{q}$ and we aim to find a first-stage and second-stage schedule with a minimum combined sum of completion times, such that at least $\Delta$ jobs share the same position in both schedules.

	We provide positive complexity results for some important special cases of this problem, as well as derive a 2-approximation algorithm to the full problem. Computational experiments examine the performance of an exact mixed-integer programming formulation and the approximation algorithm, and demonstrate the strength of a proposed polynomial time greedy heuristic.
\end{abstract}


%%%%%%%%%%%%%%%%%%%%%%%%%%%%%%%%%%%%%%%%%%%%%%%%%%%%%%%%%%%%%%%%%%%%%%%%%
%%%%%%%%%%%%%%%%%%%%%%%%%%%%%%%%%%%%%%%%%%%%%%%%%%%%%%%%%%%%%%%%%%%%%%%%%

\section{Introduction}

In this paper, we consider a recoverable robust version of a single machine scheduling problem in which $n$ jobs must be scheduled on a single machine without preemption, such that the total flow time, i.e. the sum of job completion times, is minimised. Under the $\alpha|\beta|\gamma$ scheduling notation introduced by \cite{graham1979optimization}, the nominal problem is denoted as $1||\sum C_i$. In reality, job processing times are usually subject to some degree of uncertainty. When this is the case, it is important to account for this uncertainty already in the construction of scheduling solutions. Here, we consider the case where the job processing times are assumed to lie within specified intervals and examine the recoverable robust optimisation that arises from assuming a two-stage decision process.

The nominal problem can be stated as follows. Given a set of jobs $N=\{1,\dots,n\}$ with processing times $\bm{p}=(p_1,\dots,p_n)$, find an ordering of the jobs $j\in N$ such that the sum of completion times is minimised. In other words, we want to find a permutation $\sigma$ of the set $N$, such that $\sum_{i\in N} (n+1-i)p_{\sigma(i)}$ is minimised. This nominal problem is easy to solve using the shortest processing time (SPT) rule of sorting the jobs by non-decreasing processing times. This problem can also be modelled as the following assignment problem with non-general cost structure:
\begin{align}
\min\, &\sum_{i\in N}\sum_{j\in N} p_j(n+1-i)x_{ij}\\
\text{s.t. } &\sum_{i\in N} x_{ij} = 1 & \forall j \in N \label{assignment1}\\
	     & \sum_{j\in N} x_{ij} = 1 & \forall i \in N \label{assignment2}\\
	     & x_{ij}\in\{0,1\} & \forall i,\,j \in N,
\end{align}
where $x_{ij}=1$ if job $j$ is scheduled in position $i$, and $x_{ij}=0$ otherwise.

The problem we consider in this paper considers a two-stage decision process. Suppose that in the first-stage, the jobs $j\in N$ have processing times given by $\bm{p}=(p_1,\dots,p_n)$, and in the second-stage they have different processing times given by $\bm{q}=(q_1,\dots,q_n)$. We aim to find a first-stage schedule and second-stage schedule that minimises the combined cost of these two schedules, such that the position of at least $\Delta$ jobs remain unchanged between the two schedules, i.e. only up to $n-\Delta$ jobs change position. This problem can be written as 
\begin{align*}
\tag{RecSMSP}
\min\, &\sum_{i\in N}\sum_{j\in N} p_j (n+1-i) x_{ij} + \sum_{i\in N}\sum_{j\in N} q_j (n+1-i) y_{ij} \label{eqn:recsmsp} \\
\text{s.t. } & |X \cap Y| \geq \Delta \\
    & \bm{x},\,\bm{y}\in\X
\end{align*}
where $X=\{(i,j) \in N\times N : x_{ij} = 1\}$ is the assignment corresponding to the first-stage schedule, $Y=\{(i,j)\in N\times N:y_{ij}=1\}$ is the assignment corresponding to the second-stage schedule, and $\X=\{\bm{x}\in\{0,1\}^{n\times n}: \eqref{assignment1}, \eqref{assignment2}\}$ is the set of feasible schedules. Condition $|X\cap Y| \ge \Delta$ is equivalent to demanding that
\[ \sum_{i\in N} \sum_{j\in N} x_{ij} y_{ij} \ge \Delta. \]
(\ref{eqn:recsmsp}) can be modelled as a mixed-integer program (MIP) following the linearisation of this constraint. This is achieved with the introduction of an additional set of $\bm{z}$ variables, and the following constraints:
\begin{align*}
	& z_{ij}\leq x_{ij} & \forall i,\,j \in N\\
	& z_{ij}\leq y_{ij} & \forall i,\,j \in N\\
	& \sum_{i\in N}\sum_{j\in N} z_{ij}\geq \Delta\\
	& z_{ij}\in\{0,1\} & \forall i,\,j \in N
\end{align*}
Observe how the $\bm{z}$ variables represent the shared assignments of the first and second-stages, i.e. $z_{ij}=1$ if job $j$ is assigned to position $i$ in both the first and second-stage schedules. 

This problem can be considered to be a recoverable robust optimisation problem with interval uncertainty (see \cite{liebchen2009concept}). A recoverable robust optimisation problem consist of two-stages. A full first-stage solution must be determined under the problem uncertainty, before an adversary chooses a worst-case realisation of the uncertain data. Then, in response to this realisation, the first-stage solution can be recovered, i.e. amended in some limited way, to obtain a second-stage solution. For a general survey on robust discrete optimisation problems, we refer the reader to \cite{kasperski2016robust}.

We consider exactly this recoverable robust problem in the setting where the second-stage job processing times $\bm{q}$, are uncertain, but known to lie within an interval (or box) uncertainty set given by 
$$\mathcal{U}=\left\{\bm{q}\in \mathbb{R}^n_+:q_j\in[\hat{q}_j-\bar{q}_j,\hat{q}_j+\bar{q}_j],\,j\in N\right\}.$$
That is, each job $j\in N$ has a nominal second-stage processing time given by $\hat{q}_j$, but can deviate from this nominal value by up to $\bar{q}_j$. The worst-case scenario for this uncertainty set occurs when each job simultaneously achieves its worst-case processing time, i.e. $\hat{q}_j + \bar{q}_j$ for each $j\in N$. Denoting the first-stage duration of $j\in N$ as $p_j$, and its worst-case second-stage duration as $q_j$, we get (\ref{eqn:recsmsp}).

There exists a number of papers that consider various discrete optimisation problems in the setting of recoverable robust optimisation with interval uncertainty. The resulting problems they study therefore also contain the same intersection constraints that we consider here. Our work in this paper extends this framework to the context of single-machine scheduling. \cite{busing2012recoverable} studies recoverable robust shortest path problems, whilst \cite{hradovich2017recoverable,hradovich2017recoverableb} investigate the recoverable robust spanning tree problem. \cite{kasperski2017robust} consider the recoverable robust selection problem under both discrete and interval uncertainty and, most recently, \cite{goerigk2021recoverable} consider the recoverable robust travelling salesman problem. The single-machine scheduling problem we consider is an assignment problem with a specific cost structure, and therefore \cite{fischer2020investigation} is the paper that is most closely related to our work. In this paper, the authors examine the complexity of the recoverable robust assignment problem with interval uncertainty. Amongst other results, they show that this problem is W[1]-hard with respect to $\Delta$ and $n-\Delta$. Even though they use only 0-1 costs for this reduction, note that this hardness result does not extend to the scheduling cost structure that we consider here. 

As a brief comment on other papers that consider robust assignment problems, \cite{dei2006robust} studies the problem in the context of discrete scenarios and \cite{pereira2011exact} considers interval uncertainty in combination with the regret criterion.

% Note that for the remainder of this paper, we consider general first-stage costs, and do not necessarily require that $q_j \ge p_j$ for each $j\in N$.

%The problem we consider in this paper is an assignment problem with a specific cost structure. Robust assignment problems were also considered in y\cite{dei2006robust} using discrete scenarios and in \cite{pereira2011exact} under interval uncertainty in combination with the regret criterion.  Most closely related to the work presented in this paper is \cite{fischer2020investigation}, where the complexity of the recoverable robust assignment problem with interval uncertainty is studied. Amongst other results, they show that this problem is W[1]-hard with respect to $\Delta$ and $n-\Delta$. Even though they use only 0-1 costs for this reduction, note that this hardness result does not extend to the scheduling cost structure that we consider here.

The single machine scheduling problem with the objective of minimising the (weighted) sum of completion times under uncertain job processing times is well-studied, particularly for the case of discrete uncertain scenarios. \cite{daniels1995robust}, \cite{kouvelis1997robust}, \cite{yang2002robust} and \cite{aloulou2008complexity} show that even for just two discrete scenarios, robust versions of this problem are NP-hard, whilst \cite{mastrolilli2013single} show that no polynomial-time approximation algorithm exists. \cite{zhao2010family} propose a cutting plane algorithm to solve the problem. More recently, \cite{kasperski2016single} considered the problem for the ordered weighted averaging (OWA) criterion, of which classical robustness is a special case, and \cite{kasperski2019risk} considered the problem for the value at risk (VaR) and conditional value at risk (CVaR) criteria.

The case of interval uncertainty has also received a lot of attention. \cite{daniels1995robust} characterised a set of dominance relations between jobs in an optimal schedule in the case of interval uncertainty. \cite{lebedev2006complexity} showed that the problem is NP-hard for regret robustness, whilst \cite{montemanni2007mixed} presented a compact MIP that was shown to be able to solve instances involving up to 45 jobs. \cite{kasperski20082} showed that the regret problem is 2-approximable when the corresponding nominal problem is solvable in polynomial time. For a survey of robust single machine scheduling for both discrete and interval uncertainty, see \cite{kasperski2014minmax}.

For robust single machine scheduling in the context of budgeted uncertainty, \cite{lu2014robust} presented an MIP and heuristic to solve the problem, before \cite{bougeret2019robust} examined its complexity. Most recently, \cite{bold2020recoverable} considered the recoverable robust problem under a different similarity measure to the one considered in this paper.

The structure and contributions of this paper are as follows. Section \ref{sec:properties} presents a number of positive results for the problem, including an efficient method for optimal scheduling when a set of jobs which must have the same first and second-stage positions is given. Section \ref{section:2approx} then presents a 2-approximation algorithm and greedy heuristic. The performance of the exact MIP formulation as well as the heuristics we present are examined experimentally in Section \ref{section:experiments}, before concluding remarks are made in Section \ref{section:conclusion}.

\section{Problem properties}
\label{sec:properties}

% Before examining this problem and some of its special cases in detail we present some simple observations. We begin by introducing the following ordering rule for the case when we are given a set of jobs which must share a position in the first and second-stage schedules.

% \subsection{Fixed choice of jobs}

We begin our analysis of recoverable robust single machine scheduling problem by considering the following question: Given a set $M\subseteq N$ of jobs that must be scheduled in the same position in both the first and the second stage, what is the best possible solution? That is, we consider the problem of calculating
\begin{align*}\tag{RecFix}
f(M) = \min\, &\sum_{i\in N}\sum_{j\in N} p_j (n+1-i) x_{ij} + \sum_{i\in N}\sum_{j\in N} q_j (n+1-i) y_{ij} \label{eqn:recfix} \\
\text{s.t. } & x_{ij} = y_{ij} & \forall i\in N, j \in M \\
	     & \bm{x},\,\bm{y}\in\X
\end{align*}
Note that a solution $(\bm{x},\bm{y})$ of~\eqref{eqn:recfix} may intersect on more jobs than just those in the set $M$. Also note that if $|M|\ge\Delta$, then the solution to this problem is feasible to (RecSMSP). Therefore, the optimal objective value of (RecSMSP) is equal to $\min \{ f(M) : |M| \ge \Delta\}$. In Algorithm~\ref{alg:eval}, we show how to calculate the value $f(M)$ in polynomial time.

\begin{algorithm}[h] 
\caption{Evaluation method for $f(M)$} \label{alg:eval}
\begin{algorithmic}[1]
\Procedure{Eval}{$\bm{p},\bm{q},M$}

\State Set $\bm{a}=(a_j)_{j \in N \setminus M}$ to be the vector of values $p_j$, $j\in N\setminus M$, sorted by non-decreasing values

\State Set $\bm{b}=(b_j)_{j \in N \setminus M}$ to be the vector of values $q_j$, $j\in N\setminus M$, sorted by non-decreasing values

\State Set $\bm{c}=(c_j)_{j \in M}$ to be the vector of values $p_j+q_j$, $j\in M$

\State Set $\bm{d}=\bm{a}+\bm{b} = (a_1+b_1,\dots, a_{N\setminus M}+b_{N\setminus M})$

\State Let $\bm{e}=(e_j)_{j\in N}$ be the vector found by concatenating vectors $\bm{c}$ and $\bm{d}$ and sorting by non-decreasing values

\State \Return $\sum_{i\in N} (n+1-i) e_i$
\EndProcedure

\end{algorithmic}
\end{algorithm}

% \noindent\textbf{Ordering rule.} Suppose we are given a set $M\subseteq N$, where $|M|\geq \Delta$, such that the jobs $j\in M$ must have the same position in the first and second-stage schedules.
% \begin{enumerate}
% 	\item Set $\bm{a}$ to be the vector of $p_j$'s such that $j\in N\setminus M$, sorted by non-decreasing values.
% 	\item Set $\bm{b}$ to be the vector of $q_j$'s such that $j\in N\setminus M$, sorted by non-decreasing values.
% 	\item Set $\bm{c}$ to be the vector of $p_j+q_j$'s such that $j\in M$.
% 	\item Compute $\bm{d}=(a_1+b_1,\dots, a_{N\setminus M}+b_{N\setminus M})$.
% 	\item Merge vectors $\bm{c}$ and $\bm{d}$ and sort by non-decreasing values.
% \end{enumerate}

The following example demonstrates the implementation of this ordering rule. Consider the data shown in Table~\ref{table:example} and suppose that $M=\{3,4\}$, i.e. jobs 3 and 4 must share the same position in the first and second-stage schedules. Then $N\setminus M=\{1,2,5\}$, i.e., jobs 1, 2 and 5 can be assigned different positions in the two schedules. The sorted first and second-stage processing times of jobs in $N\setminus M$ are $\bm{a} = (2,3,5)$ and $\bm{b} = (1,4,6)$ respectively, and the joint processing times for the jobs in $M$ are $\bm{c} = (14, 6)$. Then $\bm{d} = \bm{a} + \bm{b} = (3, 7, 11)$ and $\bm{e} = (3, 6, 7, 11, 14)$ as the merged and sorted vector processing times. This corresponds to the first-stage schedule $\sigma_1=(5,4,2,1,3)$, and the second-stage schedule $\sigma_2=(2,4,1,5,3)$. Note that, as required, jobs 3 and 4 are placed in the same position in both the first and second-stage schedules.

\begin{table}[htb]
\centering
\begin{tabular} {c  | r r r r r}
	$j$ & 1 & 2 & $\bm{3}$ & $\bm{4}$ & 5 \\
	\hline
	$p_j$ & 5 & 3 & $\bm{5}$ & $\bm{1}$ & 2 \\
	$q_j$ & 4 & 1 & $\bm{9}$ & $\bm{5}$ & 6 \\
	$p_j+q_j$ & 9 & 4 & $\bm{14}$ & $\bm{6}$ & 8 \\
\end{tabular}
\caption{Example problem data. Columns of $M=\{3,4\}$ are in bold.}
\label{table:example}
\end{table}

We now show that Algorithm~\ref{alg:eval} does indeed give an optimal solution to problem (\ref{eqn:recfix}).

\begin{theorem}\label{obs1}
Algorithm~\ref{alg:eval} calculates $f(M)$ for any $M\subseteq N$.
\end{theorem}

\begin{proof}
Let any $M\subseteq N$ be given. We denote these jobs as $M=\{j_1,\ldots,j_m\}$ with $m=|M|$. Let us first assume that we already know the set of slots $K=\{i_1,\ldots,i_m\}\subseteq N$ into which the jobs will be scheduled. We will then show how to find such a set. We further denote by $N\setminus M = \{j'_1,\ldots,j'_{n-m}\}$ and $N\setminus K=\{i'_1,\ldots,i'_{n-m}\}$ the sets of jobs and slots that are not part of $M$ and $K$, respectively. Without loss of generality, we assume that $i_1<i_2<\ldots<i_m$ and $i'_1<i'_2<\ldots<i'_{n-m}$.
The resulting problem thus decomposes into the two subproblems
\[
 \min_{\bm{z}\in\X} \ \sum_{k=1}^m \sum_{\ell=1}^m (n+1-i_k)(p_{j_\ell}+q_{j_\ell}) z_{k\ell} \\
+ \min_{\bm{x},\bm{y}\in\X} \ \sum_{k=1}^{n-m} \sum_{\ell=1}^{n-m} (n+1-i'_k)(p_{j'_\ell} x_{k\ell}+q_{j'_\ell} y_{k\ell}) \]
where with slight abuse of notation, $\X$ denotes the set of assignments with suitable dimension. Note that independent of the choice of $K$, we can find optimal solutions $\bm{z}^*$, $\bm{x}^*$ and $\bm{y}^*$. Sorting jobs $j\in M$ by non-decreasing joint processing times $p_j+q_j$ gives $\bm{z}^*$. We denote the resulting vector of joint processing times as $\bm{c}=(c_j)_{j\in M}$. Sorting jobs $j\in N\setminus M$ by non-decreasing $p_j$ gives $\bm{x}^*$, and by non-decreasing $q_j$ gives $\bm{y}^*$. We denote the resulting vector of sums of the sorted processing times as $\bm{d}=(d_j)_{j\in N\setminus M}$.

Let us fix these assignments accordingly and consider how to find set $K$. This is equivalent to assigning the positions $i_1,\ldots,i_m$ and $i'_1,\ldots,i'_{n-m}$ to slots in $N$. Hence, the value $f(M)$ is equal to
\[
\min_{\bm{\tau}\in\X} \sum_{i\in N} \sum_{\ell=1}^m (n+1- i\tau_{i\ell}) c_\ell +  \sum_{i\in N} \sum_{\ell=1}^{n-m} (n+1- i\tau_{i,\ell+m}) d_\ell
\]
% \begin{align*}
% \min &\sum_{i\in N} \sum_{\ell=1}^m (n+1- i\mathfrak{x}_{i\ell}) c_\ell +  \sum_{i\in N} \sum_{\ell=1}^{n-m} (n+1- i\mathfrak{y}_{i\ell}) d_\ell \\
% \text{s.t. } & \sum_{\ell=1}^m \mathfrak{x}_{i\ell} + \sum_{\ell=1}^{n-M} \mathfrak{y}_{i\ell} = 1 & \forall i \in N \\
% & \sum_{i\in N} \mathfrak{x}_{i\ell} = 1 & \forall \ell=1,\ldots,m \\
% & \sum_{i\in N} \mathfrak{y}_{i\ell} = 1 & \forall \ell=1,\ldots,n-m \\
% & \mathfrak{x}_{i\ell} \in\{0,1\} & \forall i\in N, \ell=1,\ldots,m \\
% & \mathfrak{y}_{i\ell} \in\{0,1\} & \forall i\in N, \ell=1,\ldots,n-m
% \end{align*}
An optimal solution to this problem can be found by assigning the positions according to non-decreasing processing times of their associated jobs, i.e., to sort the concatenated vector of processing times consisting of $\bm{c}$ and $\bm{d}$. Algorithm~\ref{alg:eval} is exactly this solution procedure.
\end{proof}

Therefore, if it is known which set of jobs $M$ with cardinality $\Delta$ must share a position in the first and second-stage schedule, the remaining problem can be solved in $O(n\log n)$ time. This immediately gives the following result.

\begin{corollary}
For a constant value of $\Delta$, (\ref{eqn:recsmsp}) can be solved in polynomial time $O(n^{\Delta+1}\log n)$. 
\end{corollary}
\begin{proof}
	This can be seen by simply enumerating the ${n\choose\Delta} \sim O(n^\Delta)$ possible ways of choosing a set of $\Delta$ jobs to share first and second-stage assignments. For each of the $O(n^\Delta)$ ways of fixing $\Delta$ jobs, the remaining problem can be solved using Algorithm~\ref{alg:eval} in $O(n\log n)$ time. Hence, the overall complexity is given by $O(n^{\Delta+1}\log n)$, i.e. polynomial for fixed $\Delta$.
\end{proof}

As observed in \cite{fischer2020investigation}, it is straightforward to see that problem (\ref{eqn:recsmsp}) can be solved in polynominal time for constant value of $\Delta$ by enumerating all possibilities for the intersection set $|X\cap Y|$. Note that whilst this approach would require us to check ${n^2\choose\Delta}$ many candidates for general cost functions, this number is reduced to ${n\choose\Delta}$ in our case.

% For the remainder of this paper, we consider the non-trivial case where $\Delta$ is part of the problem input.

% \section{Constant number of processing times} \label{section:Kconstant}

% 
% \begin{algorithm}[h] 
% \caption{Greedy heuristic for (\ref{eqn:recsmsp})} \label{greedy}
% \begin{algorithmic}[1]
% \State \nonumber \textbf{Input:} $\bm{p}$, $\bm{q}$, $\Delta$
% 
% \State Set $M=\emptyset$, $M'=\emptyset$
% 
% % \Do
% 
% \ForAll {$i \in N \setminus M $}
% 
% \State $M \leftarrow M \cup \{i\}$
% 
% \State Find $\pmb{x}$, $\pmb{y}$ using Algorithm~\ref{algo} for set $M$
% 
% \State $M' = \{(i,j) \in N\times N : x_{ij} = y_{ij} = 1\}$
% 
% \EndFor
% 
% % \doWhile{$|M'| < \Delta}
% 
% \end{algorithmic}
% \end{algorithm}

We now consider a special case of (\ref{eqn:recsmsp}) where there are a constant number $k$ of possible job processing times, i.e. $p_j,\,q_j\in \{d_1,\dots,d_k\}$, where $d_1,\dots,d_k\in \mathbb{R}$ for constant $k$. Note that this results in at most $k^2$ possible combinations of first and second-stage costs $p_j$ and $q_j$, that is $k^2$ possible job types. Furthermore, $f(M)=f(M')$ for any two choices $M,M'\subseteq N$ that contain the same number of each job type. We can conclude the following result.

\begin{corollary}
Let $p_j,\,q_j\in \{d_1,\dots,d_k\}$ for a constant value of $k$. Then, (\ref{eqn:recsmsp}) can be solved in polynomial time $O(n^{k^2})$ 
\end{corollary}

\begin{proof}
We consider the number of ways there are to choose $\Delta$ jobs from $k^2$ different job types. Observe that for all but the final job type, there are $\Delta+1$ ways of choosing up to $\Delta$ jobs of that type. Having chosen the number of jobs from the first $k^2-1$ job types, the number of jobs of the final type is simply equal to the number of jobs remaining. Hence there are $O((\Delta+1)^{k^2-1})$ ways to choose the $\Delta$ jobs. 
For each of these choices, we can use the ordering rule presented in Algorithm~\ref{alg:eval} to solve the resulting problem. Note that sorting $n$ bounded values is possible in $O(n)$ time. Hence (\ref{eqn:recsmsp}) with $k$ possible job processing times is solvable in $O((\Delta+1)^{k^2-1}\cdot n)\sim O(n^{k^2})$ time.
\end{proof}

\section{A 2-approximation algorithm} \label{section:2approx}

In this section, we present an approximate solution to (\ref{eqn:recsmsp}) and prove that this solution has a guaranteed worst-case approximation ratio of 2.

\begin{theorem}\label{thm:2approx}
A solution to the problem
\begin{align*}
f(N) = \min\,&\sum_{i\in N}\sum_{j\in N}p_j(n+1-i)x_{ij}+ \sum_{i\in N}\sum_{j\in N}q_j(n+1-i)y_{ij}\tag{UB-SMSP}\label{eqn:UB}\\
&|X\cap Y| = n\\
&\bm{x},\,\bm{y}\in \mathcal{X}
\end{align*}
provides a 2-approximation to (\ref{eqn:recsmsp}).
\end{theorem}

Observe that (\ref{eqn:UB}) is equivalent to (\ref{eqn:recsmsp}) with complete intersection of the first and second-stage solutions. It can thus be simplified to
\[
\min_{\bm{x}\in \mathcal{X}}\,\sum_{i\in N}\sum_{j\in N}(p_j+q_j)(n+1-i)x_{ij}
\]
Note that (\ref{eqn:UB}) can be solved in $O(n\log n)$ time by simply ordering the jobs $j\in N$ according to non-decreasing $p_j+q_j$ values. 

In the following, we provide a proof for Theorem~\ref{thm:2approx}. To this end, we consider a lower bound to (\ref{eqn:recsmsp}), which is provided by a solution to the problem 
\begin{equation}
\left(\min_{\bm{x}\in \mathcal{X}}\,\sum_{i\in N}\sum_{j\in N} p_j(n+1-i)x_{ij}\right)\ +\  \left(\min_{\bm{y}\in \mathcal{X}}\sum_{i\in N}\sum_{j\in N}q_j(n+1-i)y_{ij}\right),\tag{LB-SMSP}\label{eqn:LB}
\end{equation}
that is, (\ref{eqn:recsmsp}) without the intersection constraint. Suppose $\sigma_p$ and $\sigma_q$ are orderings of jobs $j\in N$ by non-decreasing first-stage costs $p_j$ and non-decreasing second-stage costs $q_j$, respectively. Then (\ref{eqn:LB}) has objective value equal to
\[ \sum_{j\in N}(p_{\sigma_p(j)}+q_{\sigma_q(j)})(n+1-j). \]
We compare the upper bound provided by (\ref{eqn:UB}) with the lower bound provided by (\ref{eqn:LB}) and refer to the values of these upper and lower bounds as $UB$ and $LB$, respectively.

By sorting items by $\sigma_p$ and $\sigma_q$ and considering the different positions of specific items between these two sortings, we get a different perspective on what constitutes an instance of (\ref{eqn:recsmsp}). Specifically, we can consider sorted processing times $p_1\le p_2 \le \ldots \le p_n$ and $q_1 \le q_2 \le \ldots \le q_n$, and a permutation $\pi$, such that the processing times $p_j$ and $q_{\pi(j)}$ belong to the same job. Figure~\ref{fig:sol_ex} provides an example of a (\ref{eqn:recsmsp}) considered in this way.

\begin{figure}[h]
\centering
\includegraphics[scale=0.65]{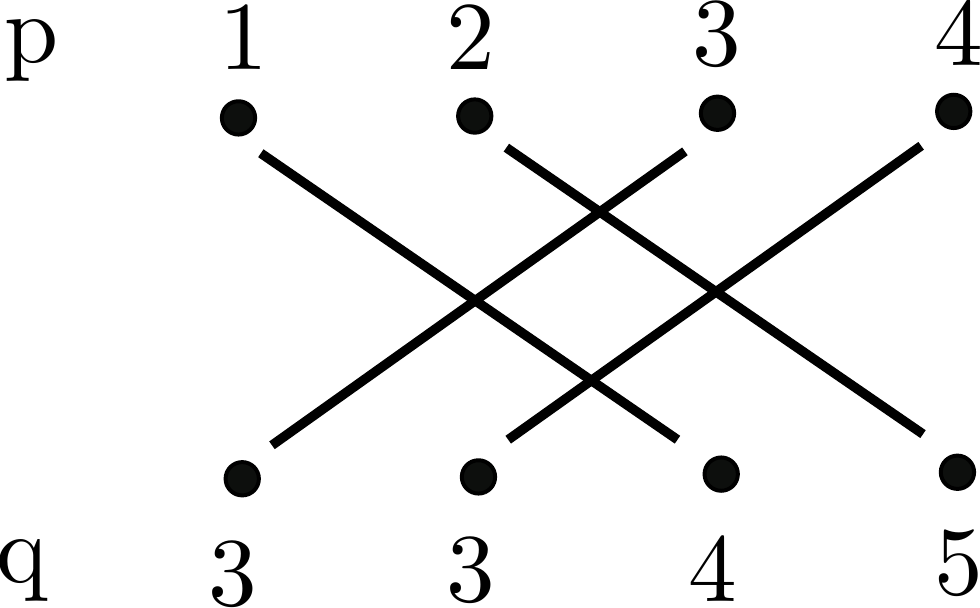}
\caption{An example instance of (\ref{eqn:recsmsp}) involving four jobs. First-stage processing times are shown in the top row and second-stage processing times are shown in the bottom row. An edge between first-stage node and second-stage node indicates that these processing times belong to a single job.}
\label{fig:sol_ex}
\end{figure}

Using this description of an instance, we begin by first proving the following structural lemma, which shows that a worst-case instance is of a form where the smallest value of $\bm{p}$ is matched with the largest value of $\bm{q}$, the second smallest value of $\bm{p}$ is matched with the second-largest value of $\bm{q}$, and so on.
 
\begin{lemma}\label{lemma:crossed}
Let any instance of (\ref{eqn:recsmsp}) be given, consisting of sorted vectors $\bm{p}$, $\bm{q}$ and a permutation $\pi$. Then the ratio $\frac{UB}{LB}$ for this instance is less or equal to the ratio for the instance where $\pi$ is replaced by $\pi'$, where $\pi'=(n,n-1,\ldots,1)$, i.e., $\pi'$ is a sorting of indices from largest to smallest. We refer to such an instance as being `fully-crossed'. 
\end{lemma}

\begin{proof}
	Matching a first-stage cost to a second-stage cost respresents a job having those respective processing times. Here, we match first and second-stage costs with the objective of maximising the ratio $\frac{UB}{LB}$. To this end, observe that $LB$ is the same for all possible matchings $\pi$ as the intersection constraint is ignored, and hence finding a matching to maximise the ratio $\frac{UB}{LB}$ is equivalent to finding a matching that maximises $UB$.

Let us consider the dual problem of $UB$, given by
\begin{align*}
\max\ & \sum_{i\in N} u_i + v_i \\
\text{s.t. } & u_j + v_i \le (n+1-i)(p_j + q_{\pi(j)}) & \forall i,j\in N 
\end{align*}
Consider any two jobs, and suppose that first and second-stage costs are labelled so that $p_1\leq p_2$ and $q_1\leq q_2$. We compare the objective values between the case when $p_1$ is matched to $q_1$ and $p_2$ is matched to $q_2$, i.e. the two jobs are uncrossed (see Figure~\ref{fig2a}), and the case when $p_1$ is matched to $q_2$ and $p_2$ is matched to $q_1$, i.e. the two jobs are crossed (see Figure~\ref{fig2b}).

\begin{figure}[htb]
	\centering
	\begin{subfigure}{.5\textwidth}
		\centering
		\includegraphics[width=.35\linewidth]{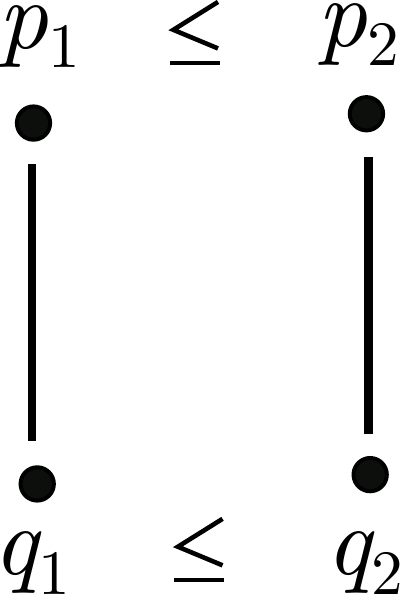}
		\caption{Uncrossed matching.}\label{fig2a}
	\end{subfigure}%
	\begin{subfigure}{.5\textwidth}
		\centering
		\includegraphics[width=.35\linewidth]{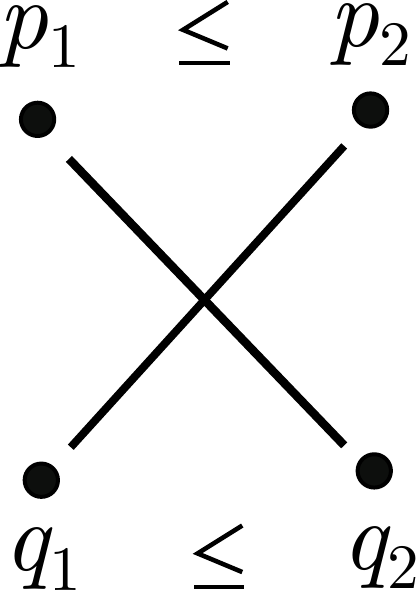}
		\caption{Crossed matching.}\label{fig2b}
	\end{subfigure}
	\caption{Uncrossed and crossed instances.}
\end{figure}

Let $(\bm{u}^*, \bm{v}^*)$ be an optimal solution to the dual of (\ref{eqn:UB}) for the instance in which the two jobs are uncrossed, and let us denote its objective value as $UB_{\text{uncrossed}}$. Note that
\begin{align*}
u^*_1 &= \min_{i\in N} \left\{ (n+1-i)(p_1+q_1)-v^*_i \right\} \\
u^*_2 &= \min_{i\in N} \left\{ (n+1-i)(p_2+q_2)-v^*_i \right\}.
\end{align*}
Now consider the dual of (\ref{eqn:UB}) for the instance in which the two jobs are crossed. We denote its objective value as $UB_{\text{crossed}}$ and construct a feasible solution $(\bm{u},\bm{v})$ by setting $\bm{v} = \bm{v}^*$,
\begin{align*}
u_1 &= \min_{i\in N} \left\{ (n+1-i)(p_1+q_2)-v^*_i \right\} \\
u_2 &= \min_{i\in N} \left\{ (n+1-i)(p_2+q_1)-v^*_i \right\}
\end{align*}
and $u_j = u^*_j$ for all other jobs $j$. Since this is a feasible solution, we have $UB_{\text{crossed}} \ge \sum_{i\in N} u_i + v_i$. 

Now for $\lambda\in[0,1]$, we consider the function $f_i(\lambda) = (n+1-i)(\lambda (p_1+q_1) + (1-\lambda) (p_2+q_2)) - v_i^*$ and define $g(\lambda) = \min_{i\in N} f_i(\lambda) + \min_{i\in N} f_i(1-\lambda)$. The minimum of concave functions is concave, and the sum of concave functions is concave, and therefore function $g$ is concave. Furthermore, it is symmetric with respect to $\lambda=0.5$. Hence, $g$ is minimised for $\lambda=0$ and $\lambda=1$, and thus for all $\lambda\in[0,1]$, $g(\lambda) \ge g(0) = g(1)$.

We therefore conclude that
\begin{align*}
UB_{\text{crossed}} - UB_{\text{uncrossed}} &\ge \sum_{i\in N} u_i + v_i - \sum_{i\in N} u^*_i + v^*_i\\
&= u_1 + u_2 - u^*_1 - u^*_2 \\
&= g(\lambda) - g(0) \ge 0
\end{align*}
where $\lambda = (p_2-p_1)/(p_2+q_2-p_1-q_1)$. Hence, by changing the matching $\pi$ as described, the value of the upper bound is not decreased. Repeating this process, we find that there is always a fully-crossed worst-case instance as claimed.
\end{proof}

We are now in a position to prove Theorem~\ref{thm:2approx}.

\begin{proof}[Proof of Theorem~\ref{thm:2approx}]
	We suppose that the first and second-stage costs are ordered so that $p_1\leq p_2\leq\dots \leq p_n$ and $q_1\leq q_2\leq \dots \leq q_n$, and that the jobs are labelled in order of their first-stage costs. Making use of Lemma \ref{lemma:crossed}, we restrict our consideration to instances where the matching of first and second-stage costs is fully crossed, i.e. job $j$ has the $j$-th largest first-stage cost $p_j$, and the $(n-j)$-th largest second-stage cost $q_{n-j}$.  

	We examine the problem of choosing values 
	$\bm{p}$ and $\bm{q}$ to maximise $\frac{UB}{LB}$, i.e.
	\begin{equation}
		\max_{(\bm{p},\bm{q})\in \mathcal{PQ}}\,\min_{\pi\in \Pi}\frac{\sum_{i\in N}(n+1-i)(p_{\pi(i)}+q_{n-\pi(i)})}{\sum_{i\in N}(n+1-i)(p_i+q_i)}, \label{wcprob}
	\end{equation}
	where $\pi\in\Pi$ is a permutation of the positions $i\in N$ such that $\pi(i)$ is the job scheduled in position $i$, and $\mathcal{PQ}=\{(\bm{p},\bm{q})\in\mathbb{R}^{n+n}_+:p_1\leq p_2 \leq \dots \leq p_n,\,q_1 \leq q_2 \leq \dots \leq q_n\}$.

	Normalising the objective function we get
	\begin{align*}
		\max_{(\bm{p},\bm{q})\in \overline{\mathcal{PQ}}}\,\min_{\pi\in \Pi}\sum_{i\in N}(n+1-i)(p_{\pi(i)}+q_{n-\pi(i)}),
	\end{align*}
	where $\overline{\mathcal{PQ}}=\{(\bm{p},\bm{q})\in\mathcal{PQ}:\sum_{i\in N}(n+1-i)(p_i+q_i)=1\}.$ For fixed $\bm{p}$ and $\bm{q}$, the inner minimisation problem over the set of permutations $\pi \in \Pi$ can be written as the following assignment problem
	\begin{align*}
		\min&\,\sum_{i\in N}\sum_{j\in N}(n+1-i)(p_j+q_{n-j})x_{ij}\\
		\textnormal{s.t.}\,&\sum_{i\in N}x_{ij}=1\quad \forall j\in N\\
				   &\sum_{j\in N}x_{ij}=1\quad \forall i\in N\\
				   &x_{ij}\in\{0,1\} \quad \forall i,\,j\in N,
	\end{align*}
	where $x_{ij}=1$ indicates that job $j$ is scheduled in position $i$, and $x_{ij}=0$ otherwise. Note that the binary constraints on the $x_{ij}$ variables can be relaxed to $x_{ij}\geq 0$ for all $i,\,j\in N$. Hence, taking the dual of this linear program, we get
	\begin{align*}
		\max\,& \sum_{i\in N}u_i + v_i\\
		\textnormal{s.t.}\,& u_j + v_i \leq (n+1-i)(p_j+q_{n-j})\quad \forall i,\,j\in N,
	\end{align*}
	and therefore the full problem~\eqref{wcprob} can be written as
	\begin{align*}
		\max\,& \sum_{i\in N}u_i + v_i\\
		\textnormal{s.t.}\,& \sum_{i\in N}(n+1-i)(p_i+q_i)=1&(\gamma)\\
				   & u_j + v_i \leq (n+1-i)(p_j+q_{n-j})\quad \forall i,\,j\in N&(\bm{x})\\
				   & p_i \leq p_{i+1} \quad \forall j\in N\setminus\{n\}&(\bm{\alpha})\\
				   & q_i \leq q_{i+1} \quad \forall j\in N\setminus\{n\}&(\bm{\beta})\\
				   & p_i,\,q_i\geq 0 \quad \forall i\in N.
	\end{align*}
	Dualising this problem (the corresponding dual variables are shown to the right of each of the above constraints) gives
	\begin{align*}
		\min&\,\gamma\\
		\textnormal{s.t.}\,&\sum_{i\in N}x_{ij}=1\quad \forall j\in N&(\bm{u})\\
				   &\sum_{j\in N}x_{ij}=1\quad \forall i\in N&(\bm{v})\\
				   &\left.\begin{aligned}
				   &n\gamma - \sum_{i\in N}(n+1-i)x_{i1} + \alpha_1\geq 0\\
				   &(n+1-j)\gamma - \sum_{i\in N}(n+1-i)x_{ij} - \alpha_j + \alpha_{j+1} \geq 0 \quad \forall j\in N\setminus\{1,n\}\hspace{18mm}\\
				   &\gamma - \sum_{i\in N}(n+1-i)x_{in} - \alpha_n\geq 0\\
					   \end{aligned}\right\rbrace &(\bm{p})\\
				   &\left.\begin{aligned}
				   &n\gamma - \sum_{i\in N}(n+1-i)x_{in} + \beta_1\geq 0\\
				   &(n+1-j)\gamma - \sum_{i\in N}(n+1-i)x_{i,n+1-j} - \beta_j + \beta_{j+1} \geq 0 \quad \forall j\in N\setminus\{1,n\}\hspace{10mm}\\
				   &\gamma - \sum_{i\in N}(n+1-i)x_{i1} - \beta_n\geq 0\\
					   \end{aligned}\right\rbrace& (\bm{q})\\
				   &x_{ij}\geq 0 \quad \forall i,\,j\in N\\
				   &\alpha_i,\,\beta_i\geq 0 \quad \forall i\in N
	\end{align*}
	(again, the corresponding primal variables are shown to the right of each of the above constraints).

	We now proceed to show that there is a feasible solution to this dual problem with an objective value of 2. Since a feasible solution to the dual problem provides an upper bound to the primal problem, finding a feasible dual solution with an objective value of 2 guarantees that the ratio $\frac{UB}{LB}$ is bounded above by 2, and therefore proves that (\ref{eqn:UB}) does indeed provide a 2-approximation to (\ref{eqn:recsmsp}).

	Setting $\alpha_i=\beta_i=0$ for all $i\in N$, the dual problem becomes
	\begin{align}
		\min&\,\gamma\nonumber\\
		\textnormal{s.t.}\,&\sum_{i\in N}x_{ij}=1\quad \forall j\in N\nonumber\\
				   &\sum_{j\in N}x_{ij}=1\quad \forall i\in N\nonumber\\
				   &(n+1-j)\gamma \geq \sum_{i\in N}(n+1-i)x_{ij}\quad \forall j\in N\tag{*}\label{eqn:constr*}\\
				   &(n+1-j)\gamma \geq \sum_{i\in N}(n+1-i)x_{i,n+1-j}\quad \forall j\in N\tag{**}\label{eqn:constr**}\\
				   &x_{ij}\geq 0 \quad \forall i,\,j\in N\nonumber.
	\end{align}
	Observe that constraint (\ref{eqn:constr**}) can be rewritten as
	\begin{equation*}
		j\gamma \geq \sum_{i\in N}(n+1-i)x_{ij}\quad \forall j\in N \label{eqn:constr**v2}.
	\end{equation*}
	Suppose that $\gamma=2$. Then constraints (\ref{eqn:constr*}) and (\ref{eqn:constr**}) enforce that
	\begin{align*}
		2(n+1-j) &\geq n+1-\sum_{i\in N}ix_{ij}\quad \forall j\in N\\
		2j &\geq n+1-\sum_{i\in N}ix_{ij}\quad \forall j\in N,
	\end{align*}
	that is 
	\begin{equation}
		\sum_{i\in N}ix_{ij}\geq\max\{2j-n-1,\,n+1-2j\} \quad \forall j\in N\label{eqn:feas}\tag{\(\dagger\)}
	\end{equation}
	By considering the above constraint, we can determine an approach for generating a feasible dual solution with objective value $\gamma=2$. This approach works as follows: take the positions $i=n,\,n-1,\dots,1$ in decreasing order, and assign jobs $j\in N$ to them, alternating between the unassigned job with the largest index and the unassigned job with the smallest index. 
	
		\begin{table}[htb]	\centering
	\begin{tabular} {c  | r r r r r r}
		$j$ & 1 & 2 & 3 & 4 & 5 & 6 \\
		$\max\{2j-n-1,\,n+1-2j\}$ & 5 & 3 & 1 & 1 & 3 & 5 \\
	\end{tabular}
	\caption{Right-hand side of constraint~(\ref{eqn:feas}) for each $j\in N$ when $n=6$.}
	\label{table:constr_ex}
	\end{table}
	
	We illustrate this approach with an example for $n=6$. Table \ref{table:constr_ex} shows constraint (\ref{eqn:feas}) for each $j\in N$. $\sum_{i\in N}ix_{ij}$ is the position in the schedule that job $j$ is assigned to, and hence these constraints enforce that jobs 1 and 6 are scheduled in the final two positions, jobs 2 and 5 are scheduled in the two positions before that, and jobs 3 and 4 are scheduled in the first two positions. Thus there are two feasible positions for each job. Following the proposed approach for generating a feasible solution, we get that $x_{6,6}=1,\,x_{5,1}=1,\,x_{4,5}=1,\,x_{3,2}=1,\,x_{2,4}=1,\,x_{1,3}=1$.
	
	For any given instance, this approach can be used to generate a $UB$ solution with objective value no worse than twice the $LB$ solution provided by (\ref{eqn:LB}).
 
\end{proof}

We next consider the value of $\frac{UB}{LB}$ as a function of $n$ and show that the 2-approximation provided by the solution (\ref{eqn:UB}) becomes tight as $n\rightarrow\infty$. 

First, consider the case in which $n$ is even. We construct a fully-crossed instance with $p_i=q_i=0$ for $i\le n/2$ and $p_i = q_i = 1$ for $i \ge n/2+1$ (which we refer to as a fully-crossed 0-1 instance). Such an instance for $n=4$ is shown in Figure \ref{fig_exa}. The value of $UB$ is given by $\sum_{i\in N}(n+i-1)=\frac{(n+1)n}{2}$, whilst the value of $LB$ is given by 
$$\sum_{i\in N,\,i\geq \frac{n}{2}}2(n+1-i) = \frac{2(n/2+1)n/2}{2}=\frac{(n/2+1)n}{2},$$
and hence when $n$ is even, we have that 
$$\frac{UB}{LB}=\frac{n+1}{n/2+1}=\frac{2n+2}{n+2}\leq 2.$$

Now consider the case in which $n$ is odd. We slightly change the construction of this instance by setting $p_i=q_i=0$ for $i\le (n-1)/2$ and $p_i = q_i = 1$ for $i \ge (n+1)/2+1$ and $p_i = 0$, $q_i=1$ for $i=(n+1)/2$. This instance, for $n=5$, is shown in Figure \ref{fig_exb}. Again, the value of $UB$ is given by $\sum_{i\in N}(n+i-1)=\frac{(n+1)n}{2}$. However, when $n$ is odd, the value of $LB$ is given by
$$\frac{n+1}{2}+\sum_{i\in N,\,i\geq \frac{n-1}{2}}2(n+1-i) = \frac{n+1+2(\frac{n-1}{2}+1)(\frac{n-1}{2})}{2} =\frac{(n+1)^2}{4},$$
and therefore, when $n$ is odd, we have that 
$$\frac{UB}{LB} =\frac{2n(n+1)}{(n+1)^2}=\frac{2n}{n+1}\leq 2.$$
Thus, as $n\rightarrow\infty$, $\frac{UB}{LB}\rightarrow 2$.

We continue this analysis by looking at the range in which the true approximation ratio lies, as a function of $\Delta/n$. In Figure~\ref{fig:approx_ratio}, the line in blue shows the upper bound of 2 as stated by Theorem~\ref{thm:2approx}. The line in orange shows a lower bound of the approximation ratio, calculated as the actual approximation ratio for a specific instance, namely, a fully-crossed 0-1 instance with $n=100$. The true approximation ratio is known to lie between these upper and lower bounds, where we already know that for $\Delta=0$, the ratio 2 is tight, and for $\Delta=n$, (UB-SMSP) provides an optimal solution, i.e. the approximation guarantee becomes 1.

\begin{figure}[htb]
\centering
\includegraphics[scale=0.75]{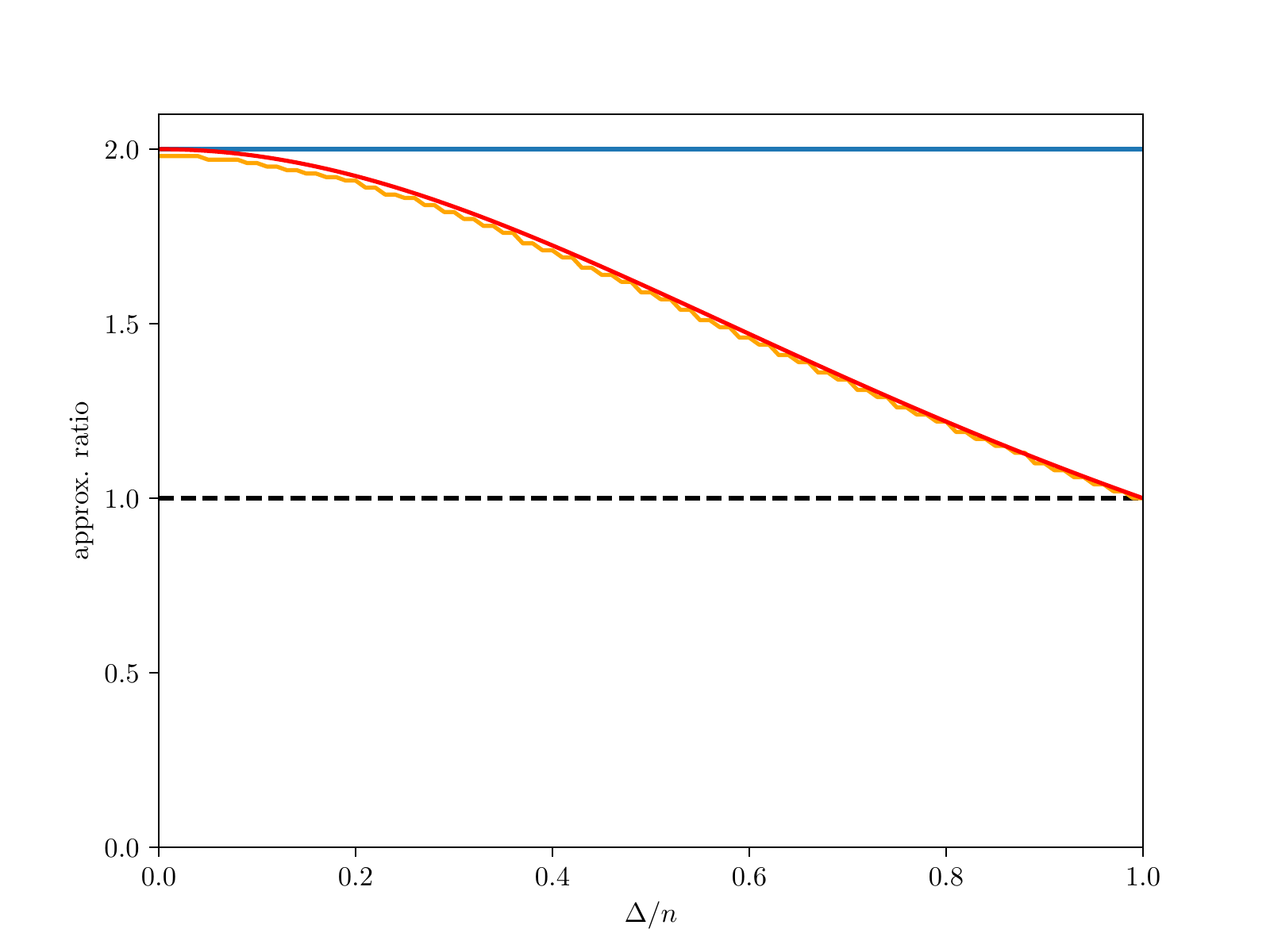}
\caption{The upper bound on the approximation ratio of (UB-SMSP), as proved in Theorem~\ref{thm:2approx}, is shown in blue. A lower bound on the approximation ratio as a function of $\Delta/n$, computed using a fully-crossed 0-1 instance with $n=100$ is shown in orange. The lower bound as a function of $\Delta/n$ for fully-crossed 0-1 instances as $n\rightarrow \infty$ is shown in red.}
\label{fig:approx_ratio}
\end{figure}

Observe that an optimal solution to a fully-crossed 0-1 instance will always use the available $\lfloor\frac{n-\Delta}{2}\rfloor$ swaps on the outer-most jobs. See Figure~\ref{fig:example} for an illustration of this. Knowing this, it becomes straightforward to compute the objective function of an optimal solution to such an instance. For example, in the case where $n-\Delta$ is even, the optimal objective value, $v^*$, is given by
\begin{align*}
	v^* =& \underbrace{\frac{2\Big(\frac{n-\Delta}{2}+1\Big)\frac{n-\Delta}{2}}{2}}_{\text{outer $(n-\Delta)/2$ swaps}} + \underbrace{\frac{\bigg(\left(\frac{n-\Delta}{2}+\Delta\right) + \left(\frac{n-\Delta}{2}+1\right)\bigg)\Delta}{2}}_{\text{inner $\Delta$ fixed jobs}}\\
		=&\frac{n^2+\Delta^2+2n}{4}.
	\end{align*}
Therefore, when $n-\Delta$ is even, the true approximation ratio can be computed as
	\begin{align*}
		\frac{UB}{v^*}=&\frac{2n(n+1)}{n^2+\Delta^2+2n}\\
		=&\frac{2n^2+2n}{(1+\gamma^2)n^2+2n}\rightarrow \frac{2}{1+\gamma^2}
	\end{align*}
as $n\rightarrow \infty$, where $\gamma=\Delta/n$. Note that a very similar analysis arrives at the same result for the case where $n-\Delta$ is odd. In Figure \ref{fig:approx_ratio}, we plot this limiting curve in red. Looking at the plot we see that for $\Delta = 0$, when the first and second solutions require no intersection, the upper and lower bounds converge as $n\rightarrow \infty$, confirming that the 2-approximation is tight. For $\Delta>0$, it remains possible that still stronger guarantees can be found.

\begin{figure}
	\centering
	\includegraphics[scale=0.5]{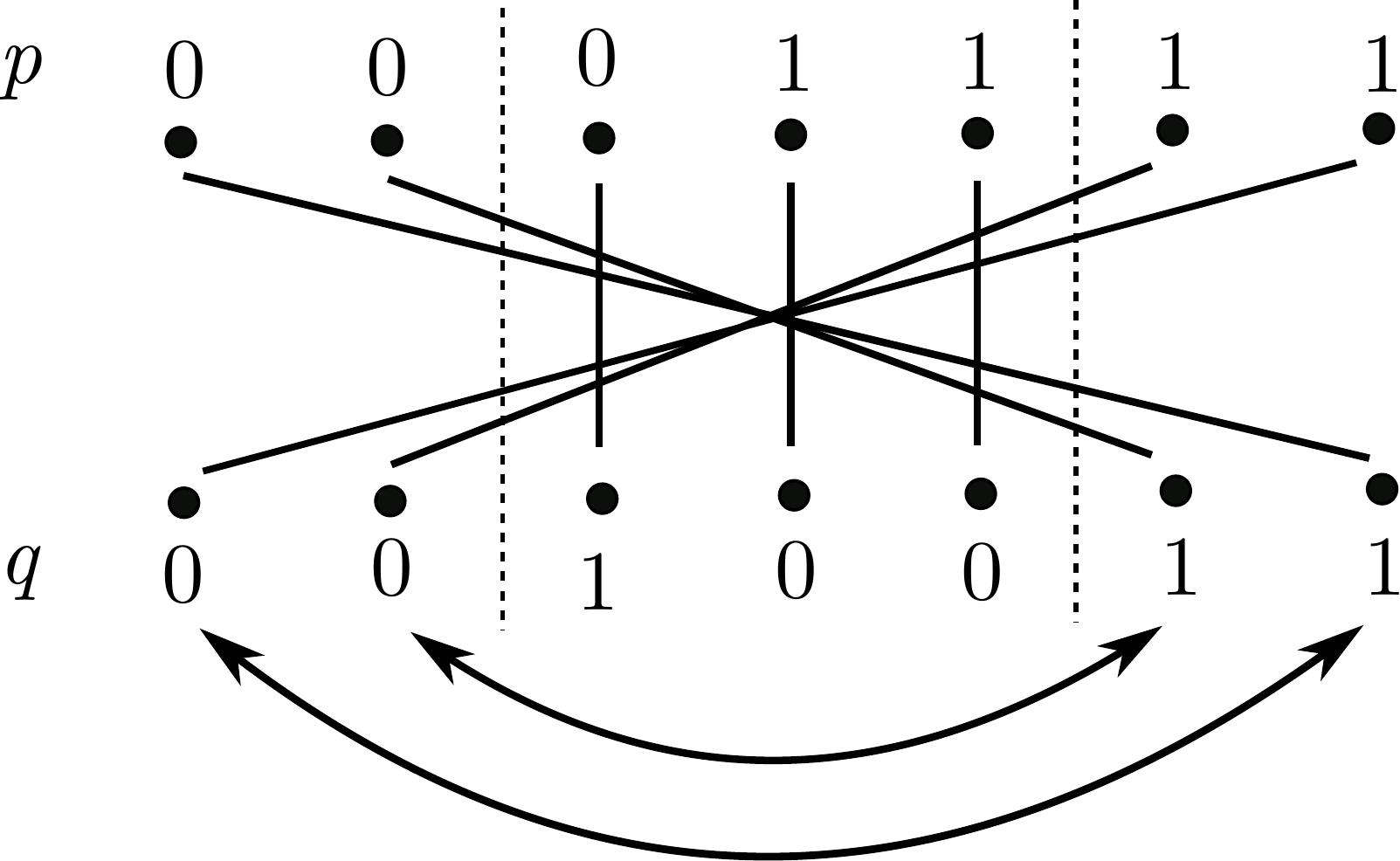}
	\caption{Optimal solution for a fully-crossed 0-1 instance where $n=7$ and $\Delta=3$. The $\lfloor\frac{n-\Delta}{2}\rfloor=2$ available swaps are applied to the outer-most jobs. This structure is shared by the optimal solutions to all fully-crossed 0-1 instances.}
	\label{fig:example}
	\end{figure}

(\ref{eqn:UB}) is found by forcing the first and second-stage schedules to be the same. Clearly, this approach is overly conservative, and in practice it is beaten by solutions in which only a subset of jobs share a position across the two stages. The following result shows that the 2-approximation guarantee still holds for such solutions.

\begin{figure}[htb]
	\centering
	\begin{subfigure}{.5\textwidth}
		\centering
		\includegraphics[width=.65\linewidth]{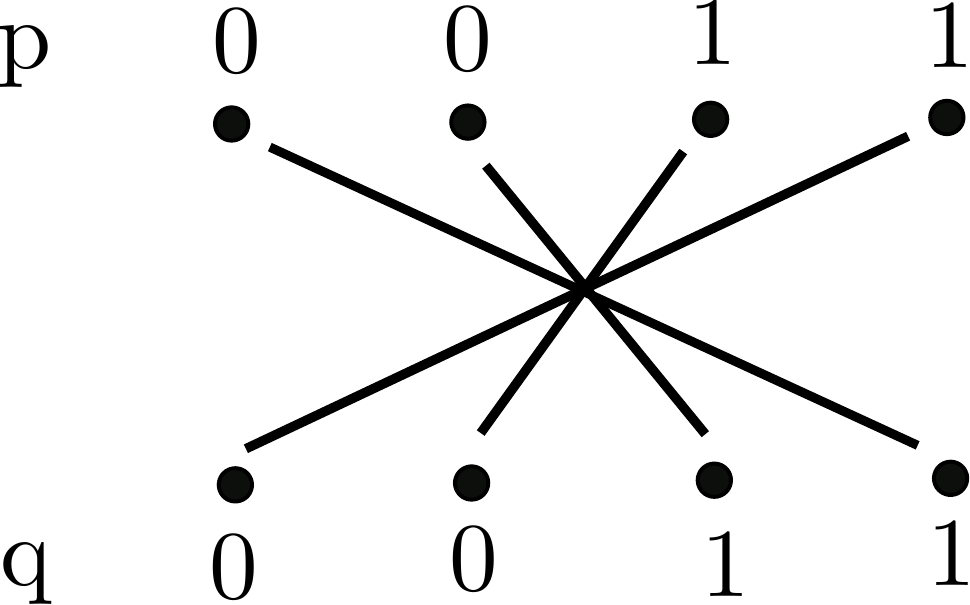}
		\caption{$n=4$}\label{fig_exa}
	\end{subfigure}%
	\begin{subfigure}{.5\textwidth}
		\centering
		\includegraphics[width=.8\linewidth]{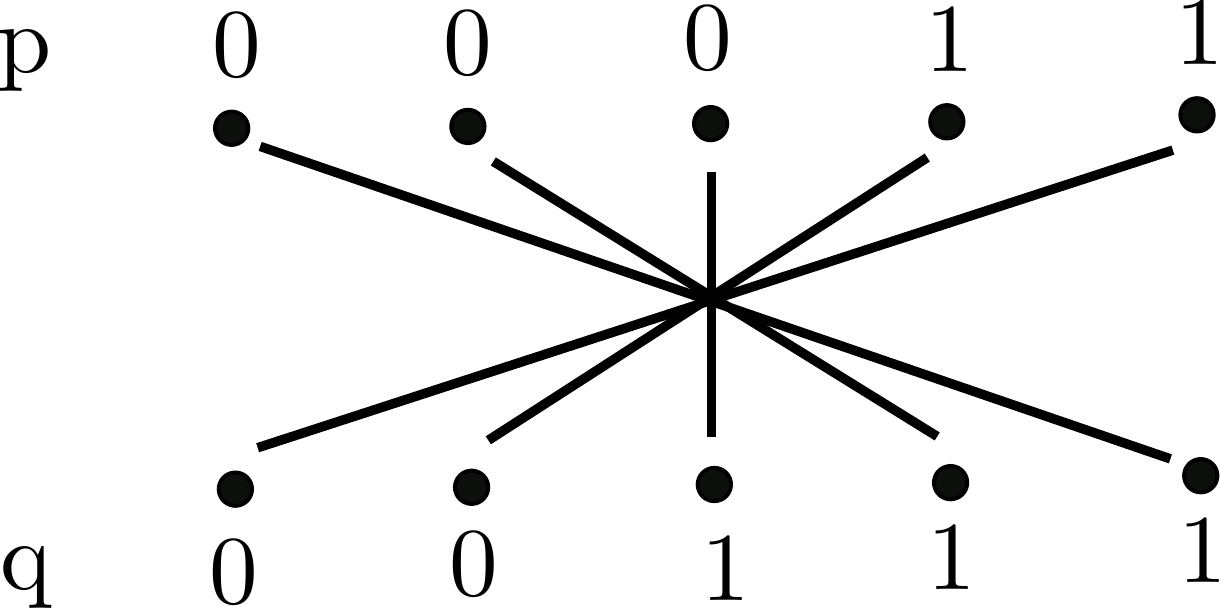}
		\caption{$n=5$}\label{fig_exb}
	\end{subfigure}
	\caption{Examples showing the constructed instances for $n=4$ and $n=5$.}
\end{figure}

\begin{lemma}\label{lem:monotone}
Function $f(M)$ is monotonically non-decreasing, i.e., $f(M) \le f(M\cup\{i\})$ for any $i\in N \setminus M$.
\end{lemma}
\begin{proof}
	The optimisation problem (\ref{eqn:recfix}) has constraints for all $j \in M$. The problem of solving (\ref{eqn:recfix}) for $M$ is therefore a relaxation of the problem of solving (\ref{eqn:recfix}) for $M\cup\{i\}$, and the claim follows.
\end{proof}

This means that any heuristic that gives a feasible pair of schedules based on Algorithm~\ref{alg:eval} is a 2-approximation as well.

\begin{corollary}\label{cor:2appx}
For any $M\subseteq N$ with $|M|\ge \Delta$, the solution generated by Algorithm~\ref{alg:eval} gives a 2-approximation for (\ref{eqn:recsmsp}) in polynomial time.
\end{corollary}

Corollary \ref{cor:2appx} suggests a natural greedy heuristic for solving (\ref{eqn:recsmsp}). This is outlined in Algorithm \ref{alg:greedy}. The procedure begins by setting $M=\emptyset$ and iteratively adds to it the element $i\in N\setminus M$ with the smallest objective value $f(M\cup\{i\})$ to the set $M$, until $|X\cap Y| \ge \Delta$. 

Due to Corollary~\ref{cor:2appx}, this greedy heuristic is ensured to be a 2-approximation. And due to Lemma \ref{lem:monotone}, we also know that its objective value can be no worse than (\ref{eqn:UB}). Each iteration of the greedy heuristic takes $O(n^2\log n)$ time and there are $O(\Delta)$ iterations, hence the heuristic runs in $O(n^3 \log n)$.
 
\begin{algorithm}[h] 
\caption{Greedy heurstic for (\ref{eqn:recsmsp})} \label{alg:greedy}
\begin{algorithmic}[1]
\Procedure {Greedy}{$\bm{p},\bm{q}, \Delta$}
\State \textbf{initialise }$M,X,Y=\emptyset$
	\While {$|X\cap Y| < \Delta$} 
 	\State Set $v=\infty$
	\For {$j^{'}\in N\setminus M$}
		\State $M^{'} \leftarrow M\cup \{j^{'}\}$
		\State $v^{'},\,X,\,Y \leftarrow \textproc{Eval}(\bm{p},\bm{q}, M^{'})$\Comment{Solve (\ref{eqn:recfix}) using \textproc{Eval}}
		\If {$v^{'} < v$}
			\State $v \leftarrow v^{'}$
			\State $j \leftarrow j^{'}$
		\EndIf
	\EndFor
	\State $M \leftarrow M \cup \{j\}$
\EndWhile
\State \Return $v,\,X,\,Y$
\EndProcedure
\end{algorithmic}
\end{algorithm}

\section{Computational experiments} \label{section:experiments}

In this section, we compare results from solving the exact MIP formulation of (\ref{eqn:recsmsp}) with those of the 2-approximation provided by solving (\ref{eqn:UB}), as well from the greedy heuristic outlined at the end of the previous section.

Before presenting these results in detail, we outline the test instances used for these experiments and the computational hardware on which these experiments were performed. Section~\ref{section:mip} investigates the performance of the MIP, focussing on its solution times and the gap to its linear relaxation. Following this, Section~\ref{section:ub} looks at the 2-approximation and shows that in practice it considerably outperforms its theoretical worst-case performance. Finally, Section~\ref{section:greedy} presents results from the greedy heuristic.

For each value of $n\in\{10,20,50,100\}$, 100 instances have been generated by randomly sampling $p_i$ and $q_i$, $i\in N$, from the set $\{1,2,\dots,100\}$. Experiments have been performed on these four instance sets for values of $\Delta\in \{0,1,2,\dots,n\}$. These problem instances, in addition to the complete results data, can be downloaded from \texttt{\url{https://github.com/boldm1/recoverable-robust-interval-SMSP}}. 

All the experiments have been run on 4 cores of a 2.30GHz Intel Xeon CPU, limited to 16GB RAM. The exact model has been solved using Gurobi 9.0.1, with a time limit of 20 minutes. The optimality gap tolerance (MIPGap) was changed to $1\times10^{-6}$; all other parameters were set to their default values.

\subsection{MIP}\label{section:mip}

We begin by considering the performance of the MIP formulation for (\ref{eqn:recsmsp}). Unsurprisingly, its performance depends on the size of $n$. However, more interestingly, it also depends on the number of free assignments across the first and second-stages, given by $n-\Delta$, and in particular whether this value is even or odd. We first present results that show the limitations of the MIP and demonstrate this dependence on the value of $n-\Delta$, before we then investigate the causes of this dependence by examining illustrative examples. 
 
Note that since $n$ is even for each of the instance sets we consider, even and odd values of $\Delta$ correspond to even and odd values of $n-\Delta$, respectively. Therefore, for ease of presentation, throughout this section we refer to the problem in terms of even and odd $\Delta$, rather than in terms of even and odd $n-\Delta$.
 
For $n=10$ and $n=20$, the MIP was solved to optimality within the 20 minute time limit for all instances and for all values of $\Delta$. For $n=50$, however, 5 instances for $\Delta=47$ and 47 instances for $\Delta=49$ were not solved to optimality within the time limit. And for $n=100$, the number of instances that were not solved to optimality within the time limit for the values of $\Delta$ are shown in the Table \ref{table:nonopt}. Note how it is only for certain odd values of $\Delta$ (and therefore odd values of $n-\Delta$) that the MIP cannot solve all instances. 

 \begin{table}[h]
	 \centering
 \begin{tabular}{l|rrrrrrrr}
	 $\Delta$ & 85 & 87 & 89 & 91 & 93 & 95 & 97 & 99\\ 
	 \hline
	 \# non-opt. & 2 & 3 & 9 & 17 & 43 & 57 & 83 & 100
 \end{tabular}
 \caption{Number of instances with $n=100$ for which the MIP could not be solved to optimality within 20 minutes, for different values of $\Delta$. All instances were solved to optimality for the values of $\Delta$ not present in the table.}
\label{table:nonopt}
 \end{table}

 Figure \ref{fig:time} shows the average time to solve the MIP for even and odd values of $\Delta$ as a percentage of $n$, for $n\in \{10,20,50,100\}$. A log-scale has been used for clarity. These plots show that solution times scale with $n$ as expected. However, for odd values of $\Delta$ we see that solution times increase considerably for large $\Delta$, whereas when $\Delta$ is even, it has little impact on solution times.

\begin{figure}[htb]
	\makebox[\linewidth][c]{
	\begin{subfigure}[b]{.6\textwidth}
		\centering
		\includegraphics[width=1\textwidth]{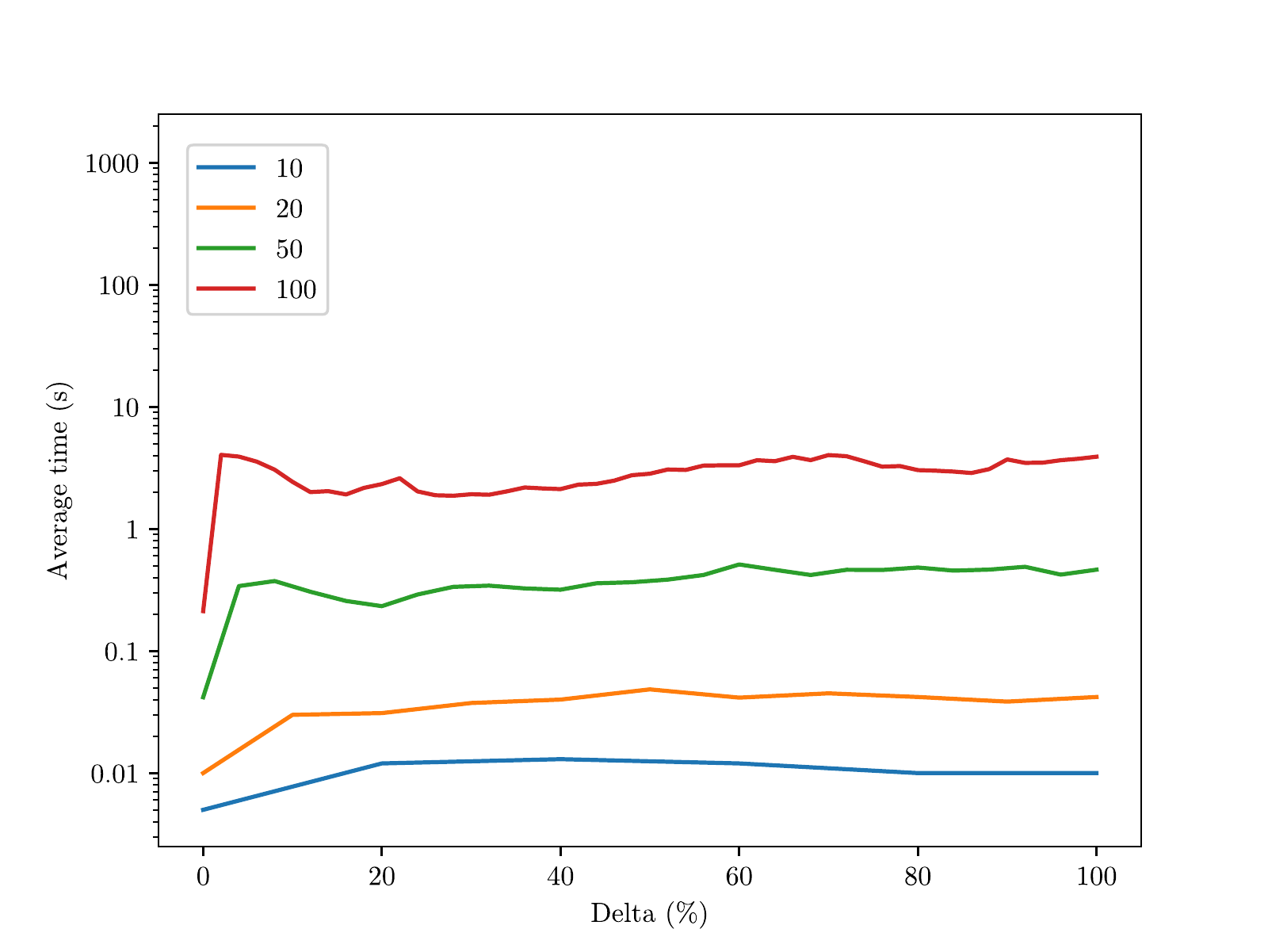}
		\caption{Even $\Delta$ values.}\label{fig:even_time}
	\end{subfigure}%
	\begin{subfigure}[b]{.6\textwidth}
		\centering
		\includegraphics[width=1\textwidth]{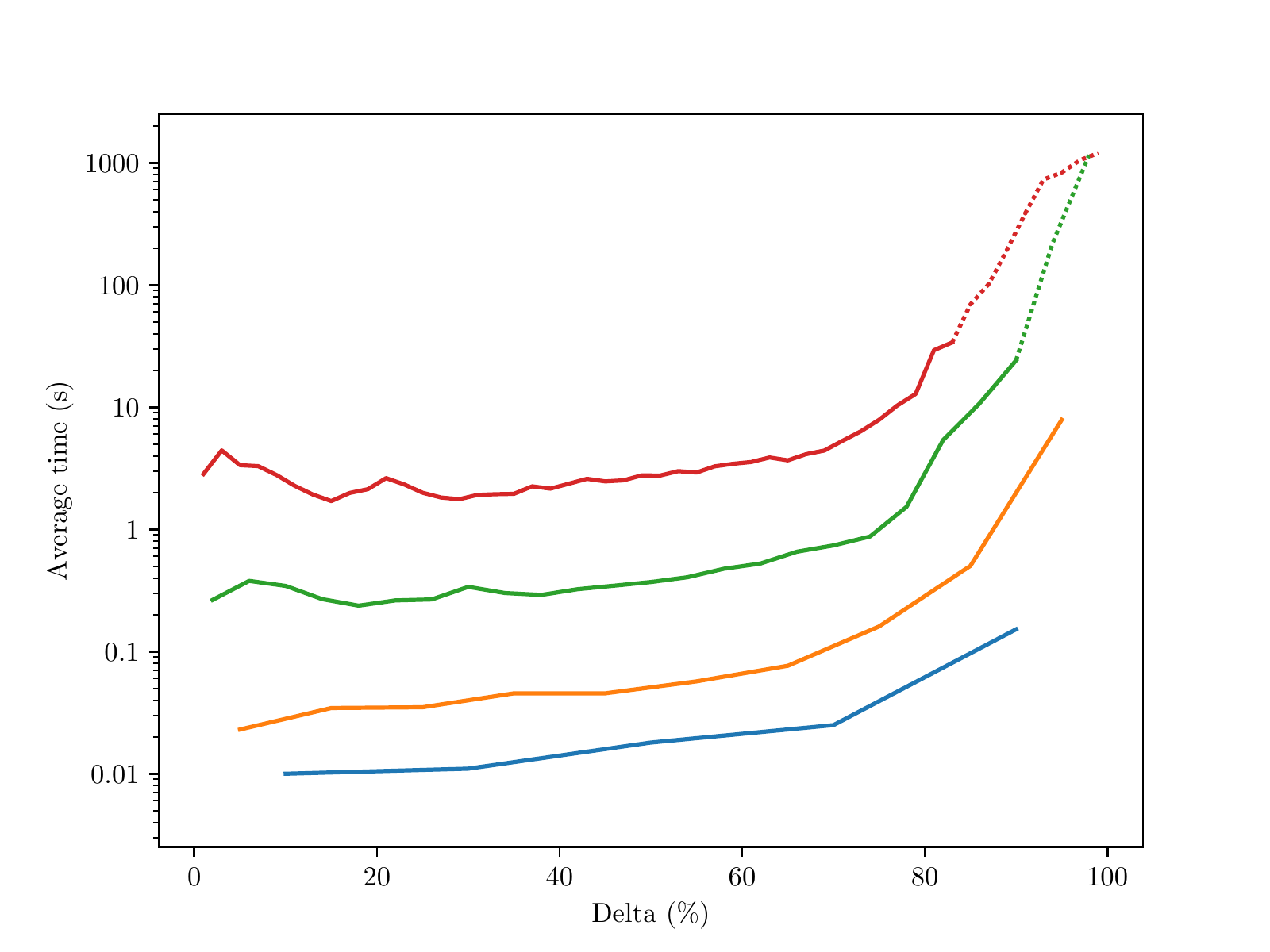}
		\caption{Odd $\Delta$ values.}\label{fig:odd_time}
	\end{subfigure}
	}
	\caption{Average time to solve the MIP for even and odd values of $\Delta$ as a percentage of $n$, for $n\in \{10,20,50,100\}$. Line becomes dotted when not all instances were solved to optimality within the 1200s (20 minute) time limit.}
	\label{fig:time}
\end{figure}

Clearly, the MIP cannot be easily solved for $n\geq 50$ for certain large values of $\Delta$, and therefore the performance of the approximate solution approaches we propose are of genuine interest. Before we examine their performances over the following two sections, we investigate the marked difference in difficultly between instances with even and odd $\Delta$.

We first consider the tightness of the MIP formulation by looking at the average gap between its solutions and the solutions of its linear relaxation (i.e. the average LP gap) in Figure \ref{fig:gap}. Observe that the average LP gaps for instances where $\Delta$ is even are orders of magnitudes smaller than for instances where $\Delta$ is odd; the largest single LP gap for an even-$\Delta$ instance and an odd-$\Delta$ instance for each of the instance sets are given in Table \ref{table:lpgap}.

 \begin{table}[H]
	 \centering
 \begin{tabular}{l|rcrc}
	 & \multicolumn{4}{c}{LP gap (\%)}\\
	 $n$ & \multicolumn{2}{c}{even-$\Delta$} & \multicolumn{2}{c}{odd-$\Delta$} \\ 
	 \hline
	 10 & 0.20 & ($\Delta=4$) & 5.66 & ($\Delta=9$) \\
	 20 & 0.17 & ($\Delta=14$) & 2.89 & ($\Delta=19$) \\
	 50 & 0.07 & ($\Delta=38$) & 1.77 & ($\Delta=49$) \\
	 100 & 0.03 & ($\Delta=82$) & 0.89 & ($\Delta=99$)
 \end{tabular}
 \caption{Largest LP gap for a single even-$\Delta$ instance and odd-$\Delta$ instance for each of the instance sets.}
\label{table:lpgap}
 \end{table}

Figure \ref{fig:non_opt} shows the percentage of instances for which the LP gap is non-zero, for even and odd values of $\Delta$ as a percentage of $n$, for $n\in \{10,20,50,100\}$. Looking at this figure, we see that the proportion of instances with non-zero LP gaps is broadly similar for each of the four instance sets over the values of $\Delta$. Therefore, to explain the results seen in Figure \ref{fig:gap}, it must be the case that the LP gaps become smaller as $n$ increases. Figure \ref{fig:non_opt} also shows that for larger values of $\Delta$, the frequency of non-zero LP gaps increases dramatically for instances with odd values of $\Delta$, but not for instances with even values of $\Delta$.
 
\begin{figure}[htbp]
% 	\makebox[\linewidth][c]{
	\begin{subfigure}[b]{.5\textwidth}
		\centering
		\includegraphics[width=1\textwidth]{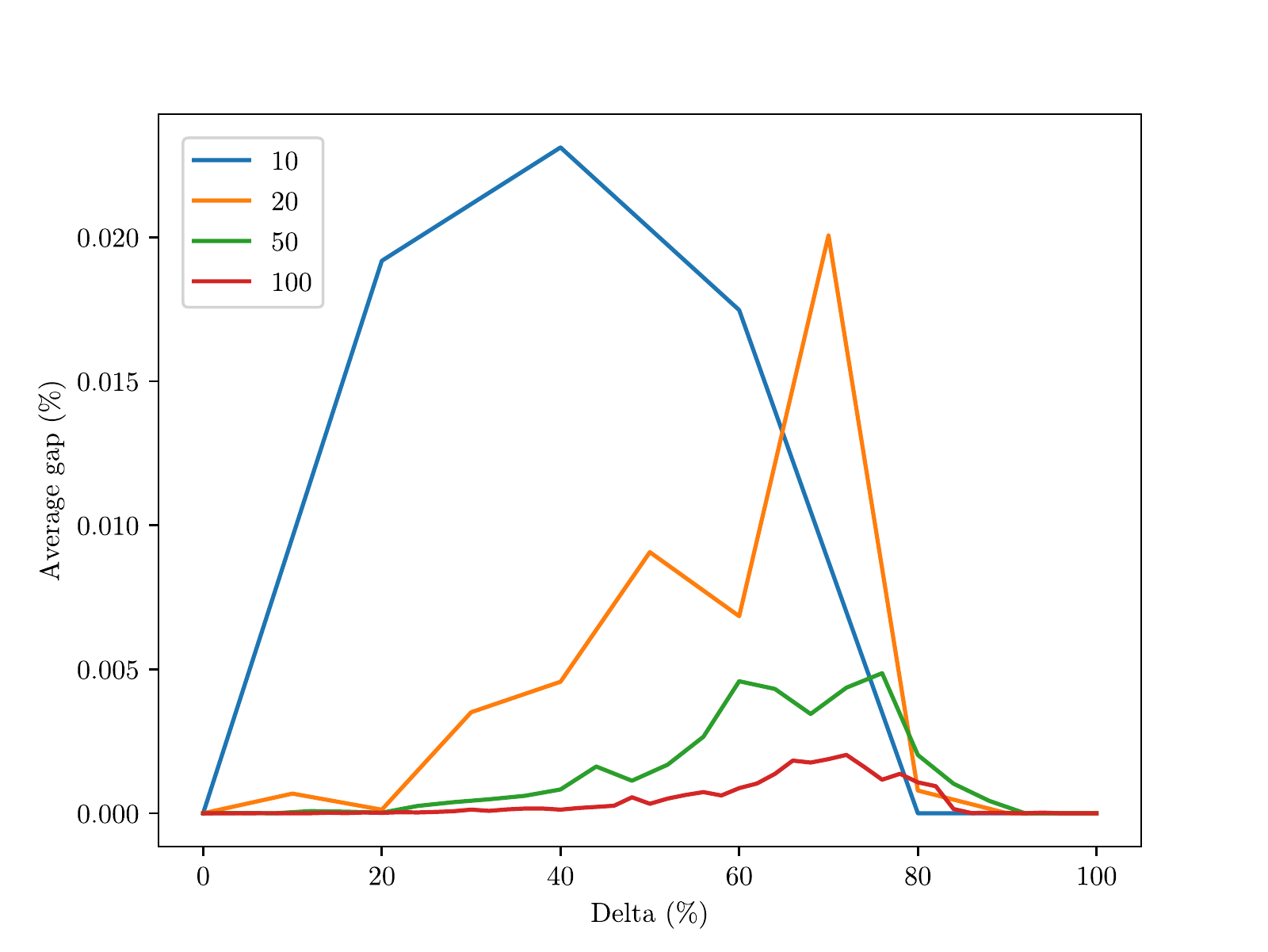}
		\caption{Even $\Delta$ values.}\label{fig:even_gap}
	\end{subfigure}%
	\begin{subfigure}[b]{.5\textwidth}
		\centering
		\includegraphics[width=1\textwidth]{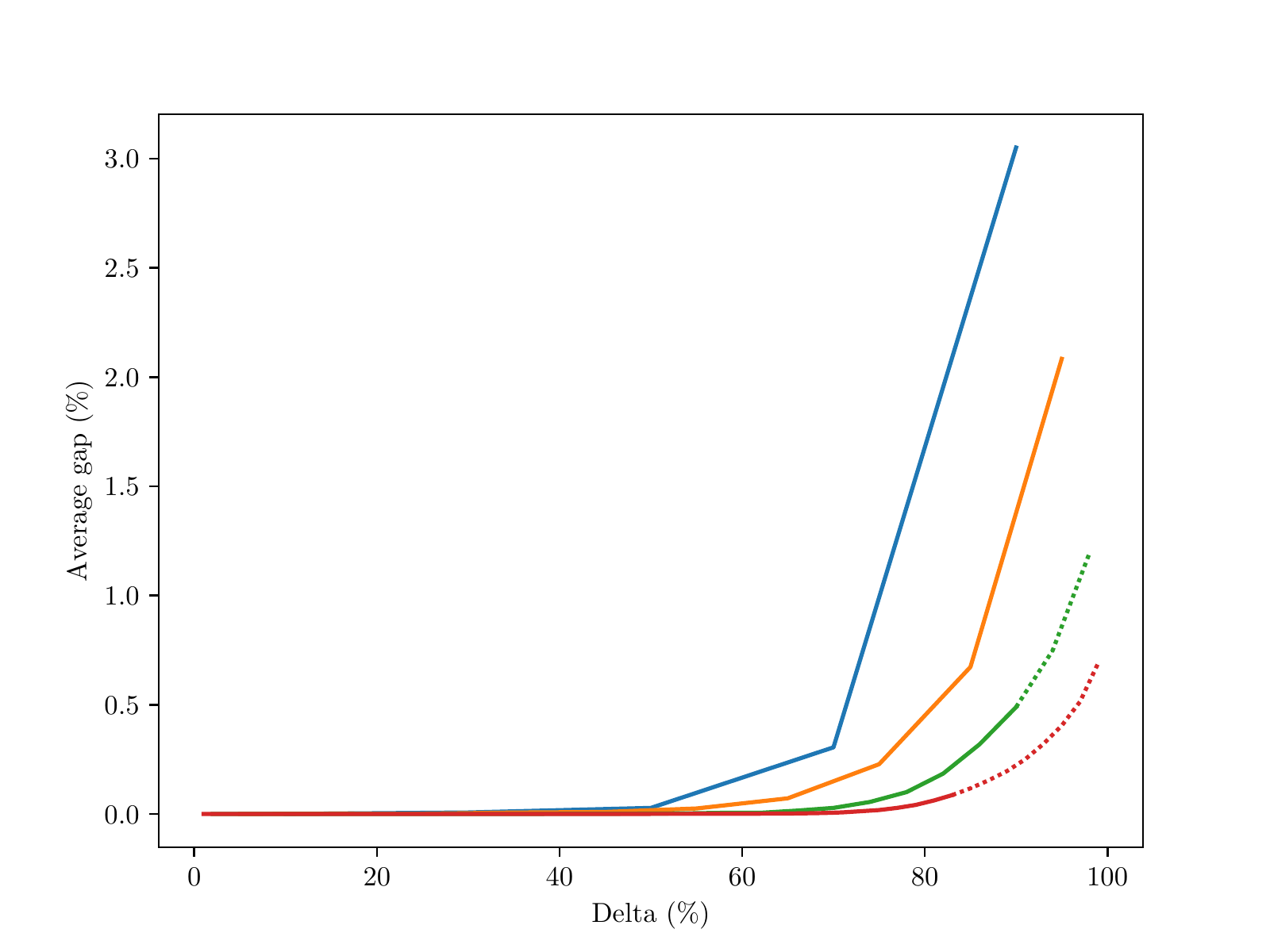}
		\caption{Odd $\Delta$ values.}\label{fig:odd_gap}
		\end{subfigure}
% 	}
	\caption{Average LP gap for even and odd values of $\Delta$ as a percentage of $n$, for $n\in \{10,20,50,100\}$. Line becomes dotted when not all instances were solved to optimality within the 20 minute time limit.}
	\label{fig:gap}
% 	\makebox[\linewidth][c]{
\end{figure}

\begin{figure}[htbp]
	\begin{subfigure}[b]{.5\textwidth}
		\centering
		\includegraphics[width=1\textwidth]{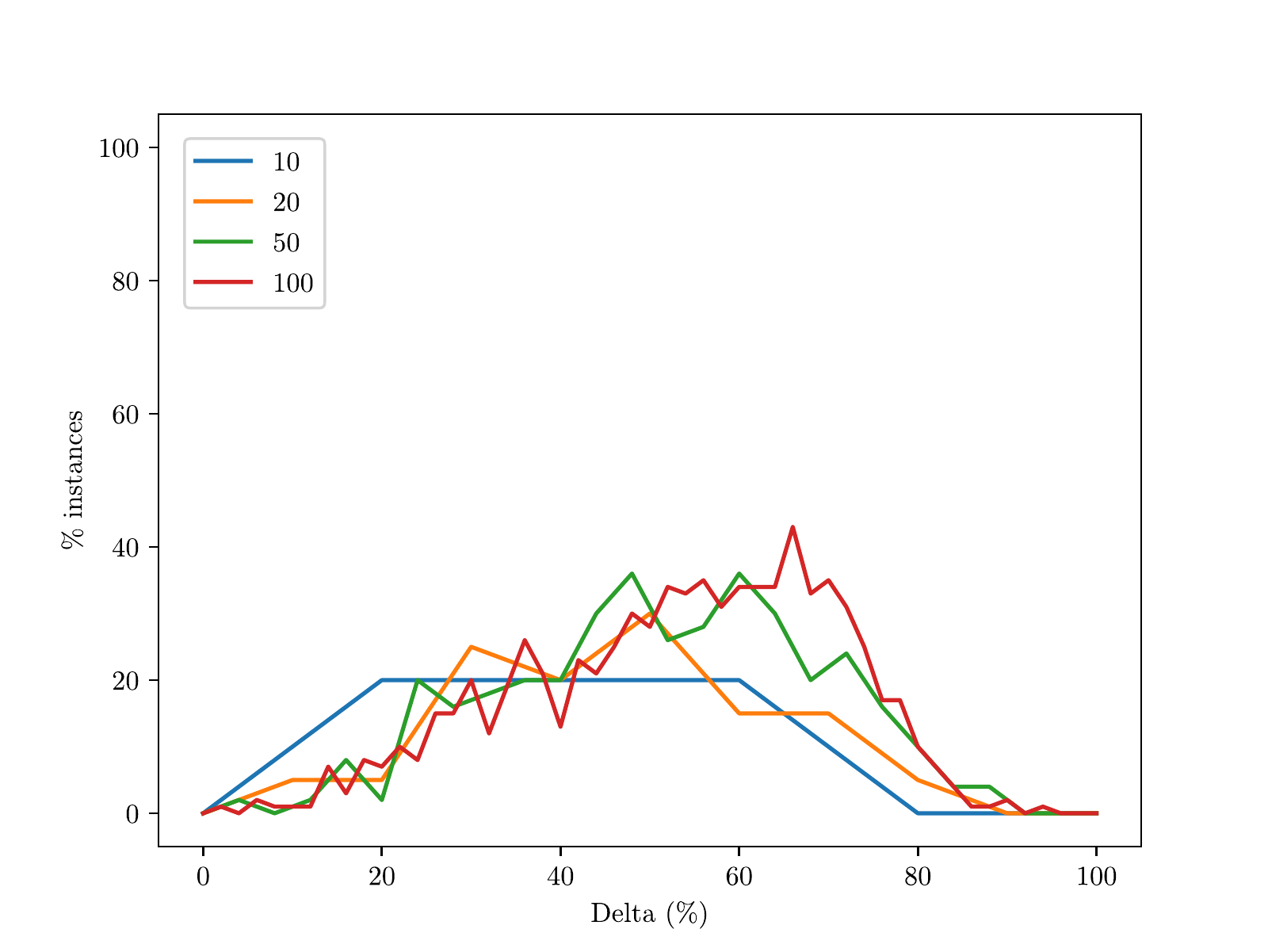}
		\caption{Even $\Delta$ values.}\label{fig:even_nonopt}
	\end{subfigure}%
	\begin{subfigure}[b]{.5\textwidth}
		\centering
		\includegraphics[width=1\textwidth]{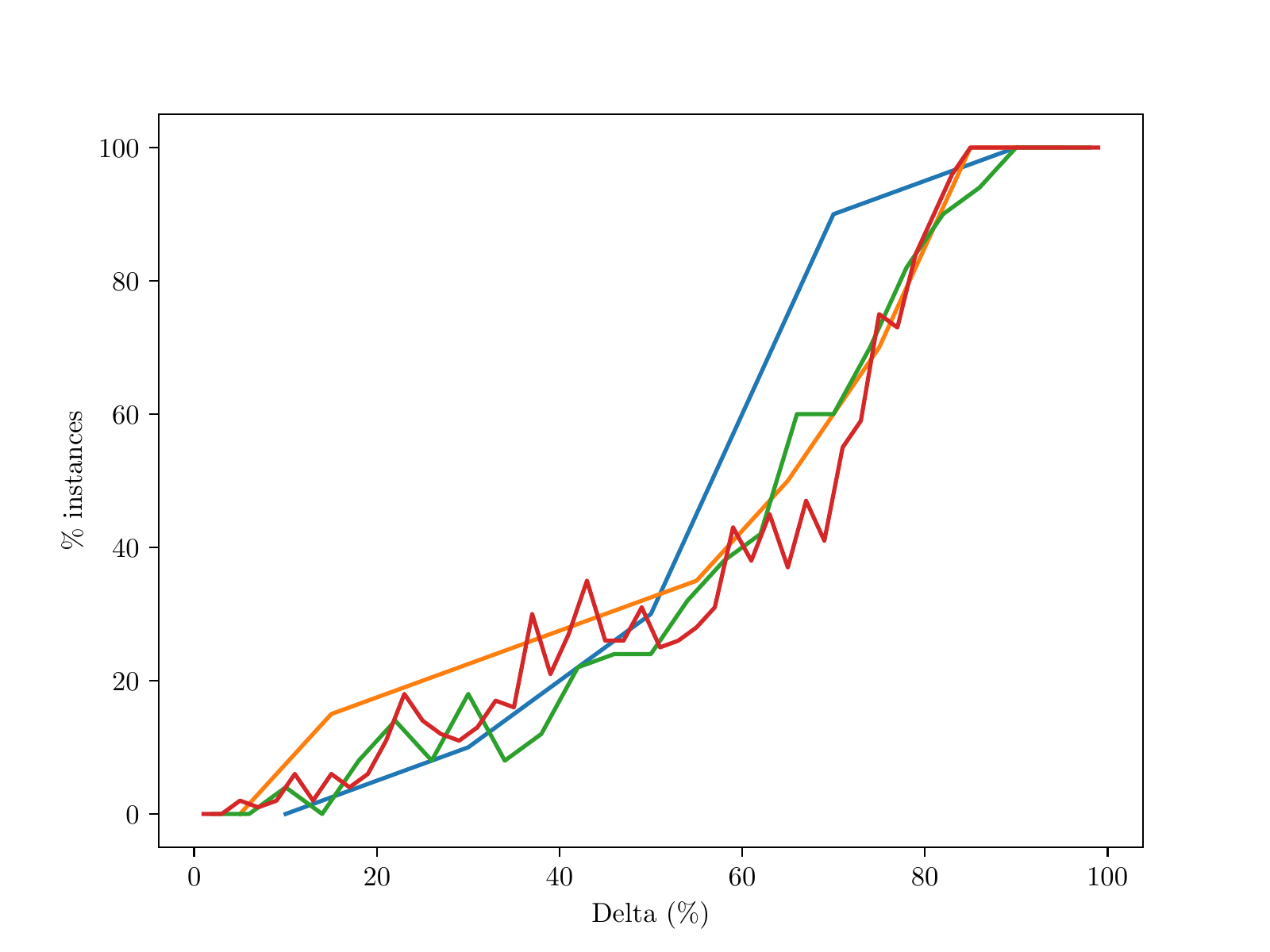}
		\caption{Odd $\Delta$ values.}\label{fig:odd_nonopt}
		\end{subfigure}
% 	}
	\caption{Percentage of instances for which the LP gap is non-zero, for even and odd values of $\Delta$ as a percentage of $n$, for $n\in \{10,20,50,100\}$.}
	\label{fig:non_opt}
\end{figure}

%\begin{figure}
%	\makebox[\linewidth][c]{
%	\begin{subfigure}[b]{.6\textwidth}
%		\centering
%		\includegraphics[width=1\textwidth]{n_nonopt_even.pdf}
%		\caption{Even $\Delta$ values.}\label{fig:even_nonopt}
%	\end{subfigure}%
%	\begin{subfigure}[b]{.6\textwidth}
%		\centering
%		\includegraphics[width=1\textwidth]{n_nonopt_odd.pdf}
%		\caption{Odd $\Delta$ values.}\label{fig:odd_nonopt}
%		\end{subfigure}
%	}
%	\caption{Percentage of instances for which the LP gap is non-zero, for even and odd values of $\Delta$ as a percentage of $n$, for $n\in \{10,20,50,100\}$.}.
%	\label{fig:non_opt}
%\end{figure}

\begin{figure}[htbp]
% 	\makebox[\linewidth][c]{
	\begin{subfigure}[t]{.5\textwidth}
		\centering\captionsetup{width=0.9\linewidth}
		\includegraphics[width=0.8\textwidth]{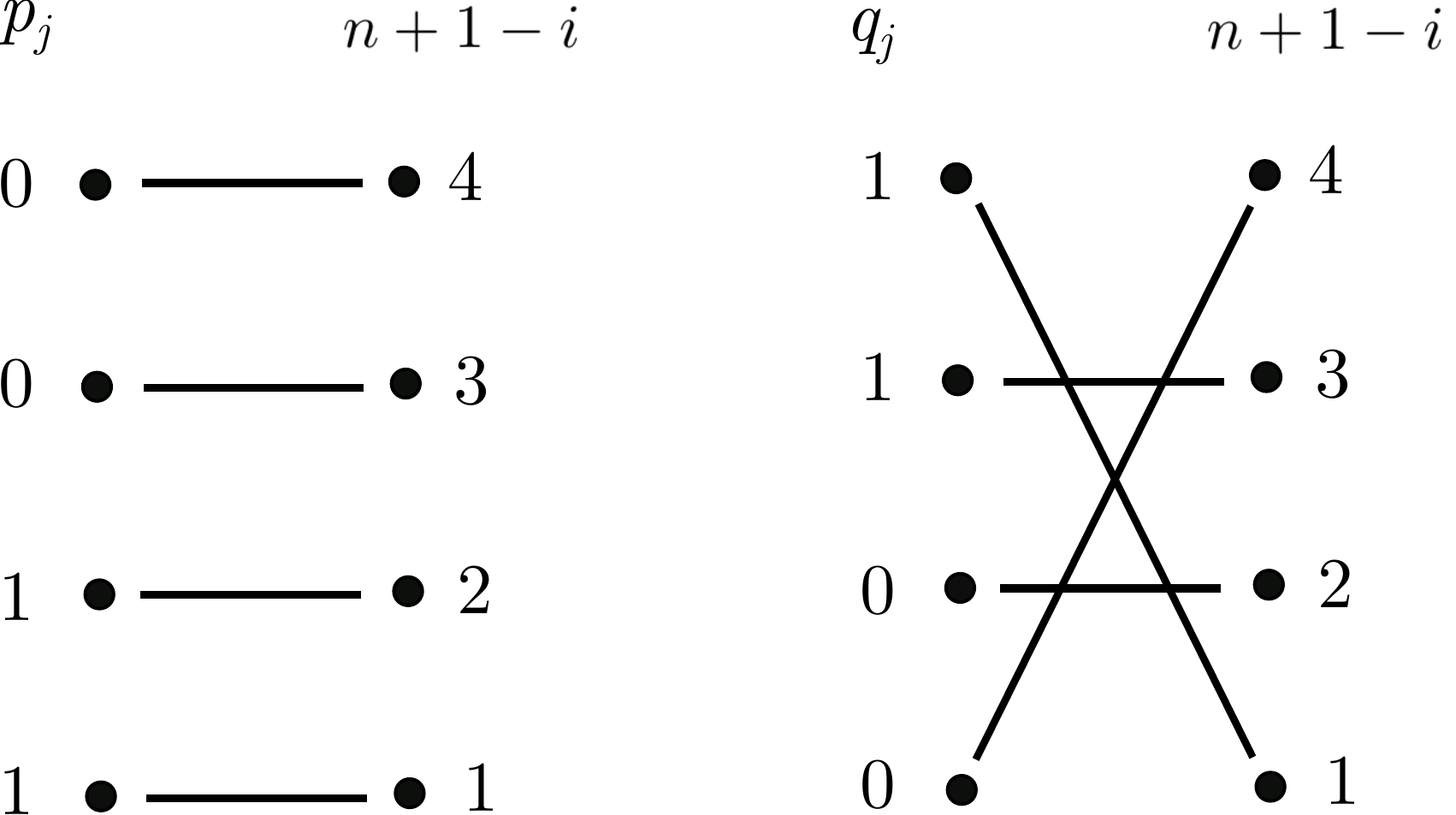}
		\caption{Solution found by MIP for $\Delta=1$ and $\Delta=2$ and by linear relaxation for $\Delta=2$. Objective value is $v=3+4=7$.}\label{fig:d2}
	\end{subfigure}%
	\begin{subfigure}[t]{.5\textwidth}
		\centering\captionsetup{width=0.9\linewidth}
		\includegraphics[width=0.8\textwidth]{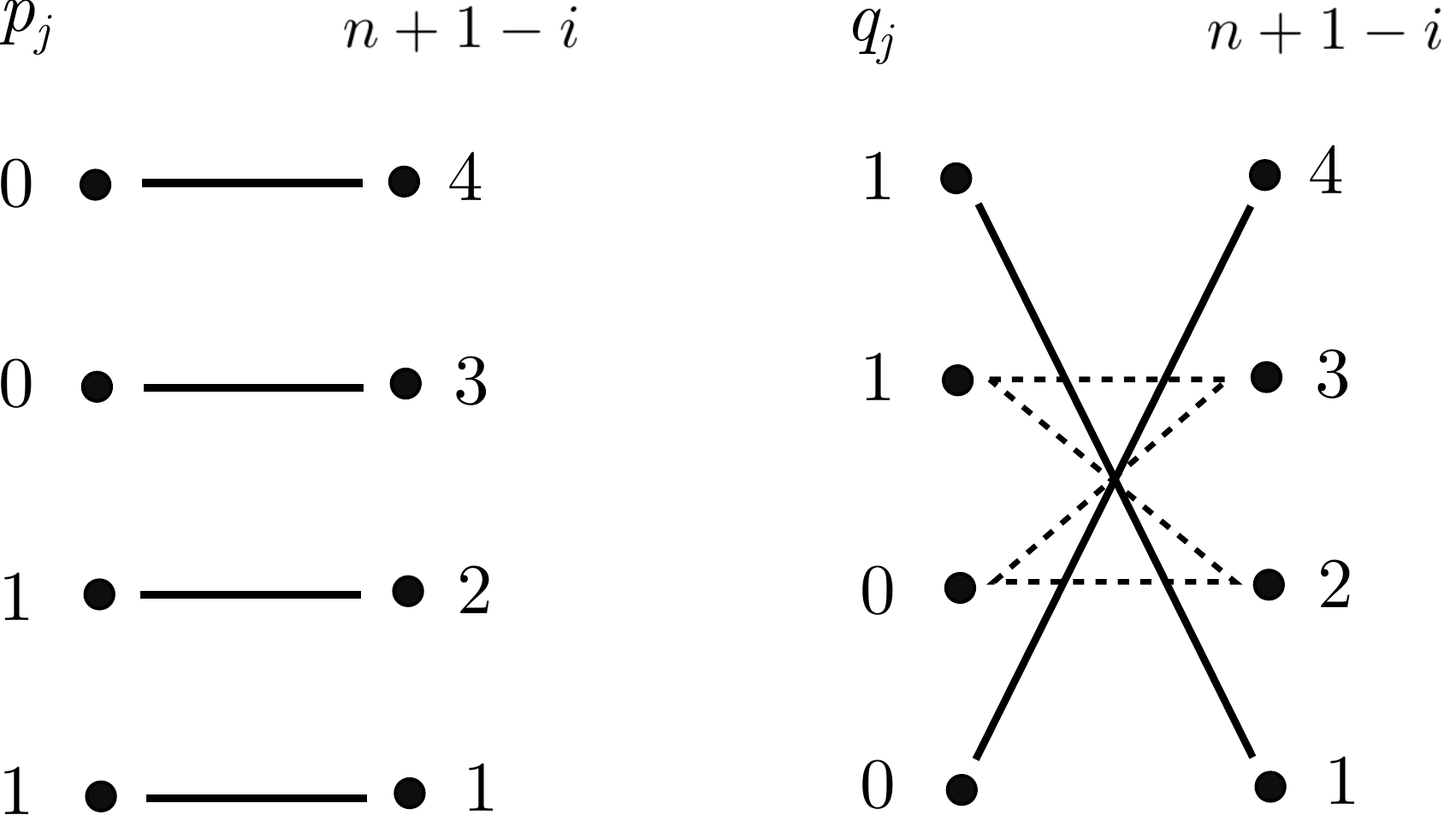}
		\caption{Solution found by linear relaxation for $\Delta=1$. Objective value is $v=3+3.5=6.5$.\vspace{4.5mm}}\label{fig:d1}
		\end{subfigure}
% % 	}
	\caption{Solutions to instance shown in Figure \ref{fig_exa}, found by MIP and its linear relaxation for $\Delta=1$ and $\Delta=2$.}
\end{figure}
\begin{figure}[htbp]
% 	\makebox[\linewidth][c]{
	\begin{subfigure}[t]{.5\textwidth}
		\centering\captionsetup{width=0.9\linewidth}
		\includegraphics[width=0.8\textwidth]{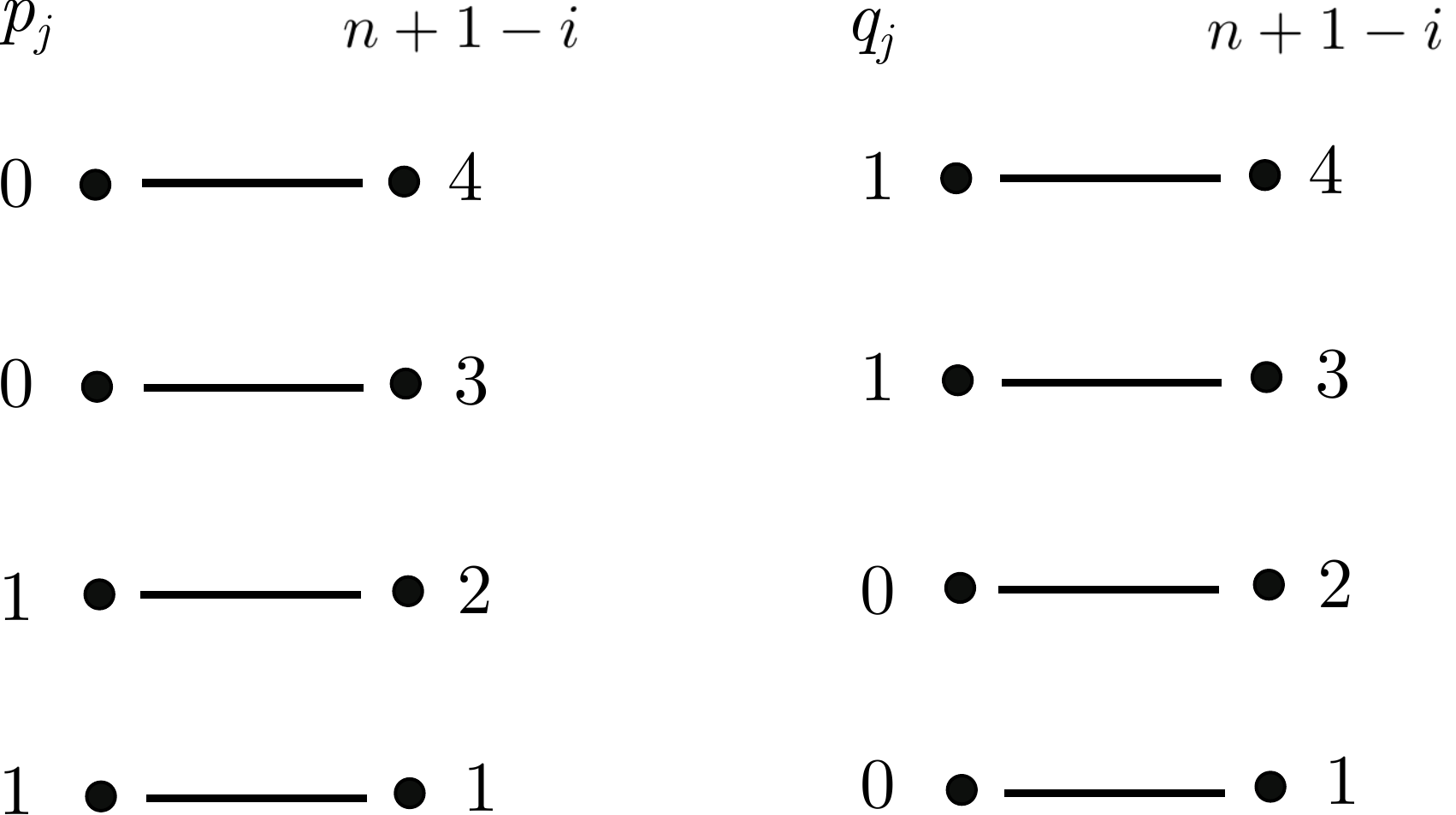}
		\caption{Solution found by MIP for $\Delta=3$ and $\Delta=4$ and by linear relaxation for $\Delta=4$. Objective value is $v=3+7=10$.}\label{fig:d4}
	\end{subfigure}%
	\begin{subfigure}[t]{.5\textwidth}
		\centering\captionsetup{width=0.9\linewidth}
		\includegraphics[width=0.8\textwidth]{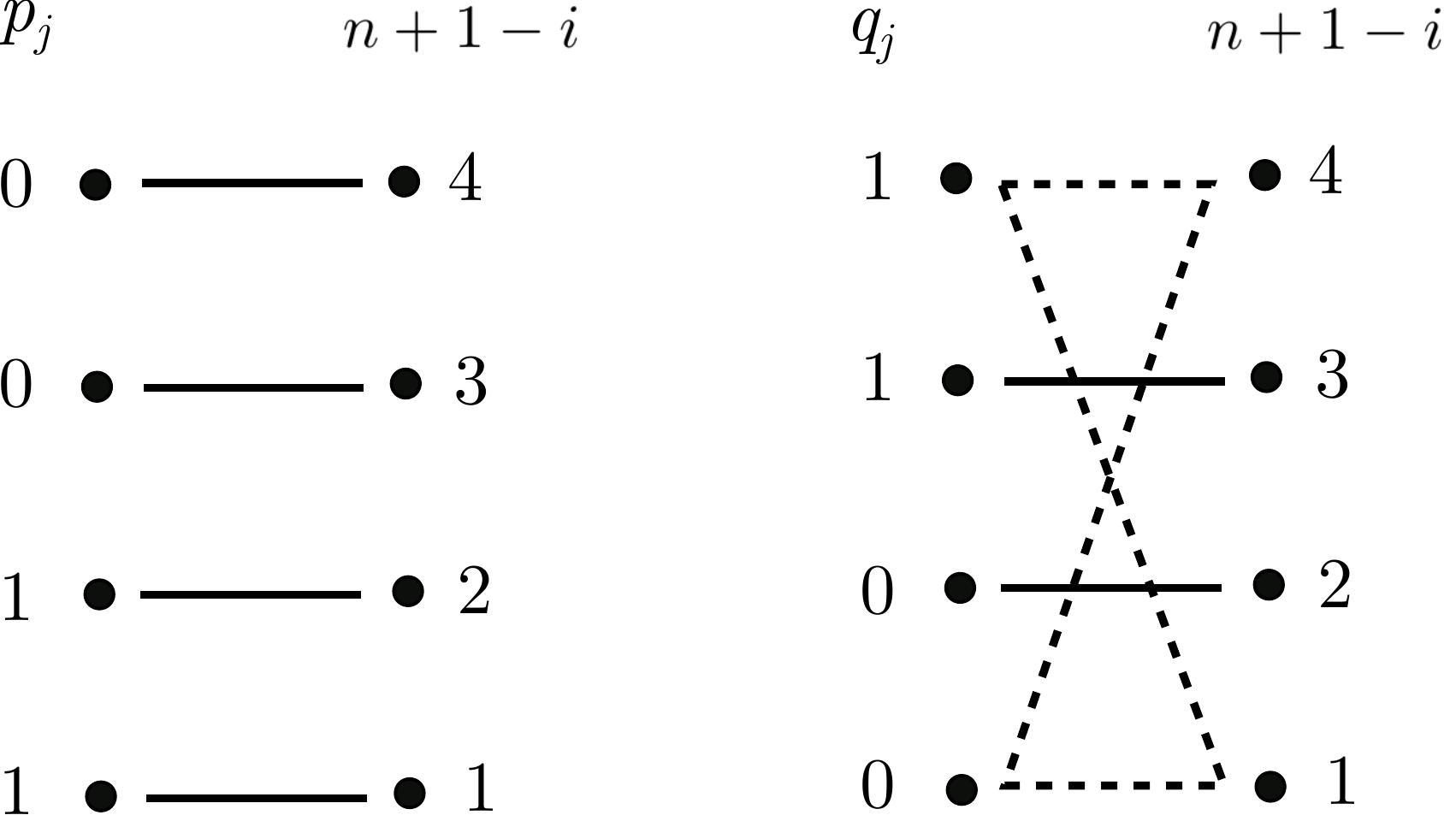}
		\caption{Solution found by linear relaxation for $\Delta=3$. Objective value is $v=3+5.5=8.5$.\vspace{4.5mm}}\label{fig:d3}
		\end{subfigure}
% 	}
	\caption{Solutions to instance shown in Figure \ref{fig_exa}, found by MIP and its linear relaxation for $\Delta=3$ and $\Delta=4$.}
	
\end{figure}

In an attempt to shed light on this distinction, we consider the example with $n=4$ jobs with binary processing times given in Figure \ref{fig_exa}. Given $\Delta$, we are free to allow up to $n-\Delta$ assignments to differ between the first and second-stage schedules. For clarity of presentation, in the following we assume that the first-stage schedule is a fixed horizontal assignment, and make changes to the assignments in the second-stage only.

%In Figure \ref{fig:ex_n4d4} we show the solution with a first-stage cost of 3 and a second-stage cost of 7. When $\Delta< n$, we are free to allow up to $n-\Delta$ assignments to differ between the first and second-stage schedules. In the following we assume that the horizontal assignment of the first-stage schedule shown in Figure \ref{fig:ex_n4d4} remains fixed, and we make changes to the assignments in the second-stage schedule only.

Observe that for a given $\Delta$, the best integral solution is found by swapping the positions of the jobs scheduled first and last, then swapping the positions of the jobs scheduled second and second-last, and so on, until no more swaps can be made without exceeding the $n-\Delta$ available free assignments. When $n-\Delta$ is even, the integral solution is able to change the assignments of exactly $n-\Delta$ jobs, and its linear relaxation finds the same solution and gains no advantage. When $n-\Delta$ is odd however, the integral solution is only able to change the assignments of $n-\Delta-1$ jobs. In this case, the linear relaxation is able to make use of fractional assignments to beat the integral solution.

For example, when $\Delta=1$, the MIP swaps the positions of only two jobs (Figure \ref{fig:d2}). The linear relaxation gains an advantage over the integral solution since it can fractionally swap a further two jobs to make use of the last remaining free assignment (Figure \ref{fig:d1}). In this instance the LP gap is given by $\frac{7-6.5}{6.5}\approx7.7$\%. 

Similarly, when $\Delta=3$, the integral solution cannot feasibly swap the positions of any jobs (Figure \ref{fig:d4}), whereas its linear relaxation is able to fractionally swap the positions of the first and last jobs (Figure \ref{fig:d3}). The LP gap in this case is given by $\frac{10-8.5}{8.5}\approx$17.6\%.

Experimental results seem to suggest that the maximum LP gap the can be achieved is 20\%, which occurs for a fully-crossed instance with binary processing times when $n=2$ and $\Delta=1$ and when $n=3$ and $\Delta=2$.

\subsection{\ref{eqn:UB}}\label{section:ub}

 As the results in the previous section demonstrate, large instances of (\ref{eqn:recsmsp}) cannot be solved within a reasonable time limit by the MIP. Therefore, accurate heuristic approximations are valuable for solving this problem. In this section, we look at the quality of the solutions to (\ref{eqn:UB}) in practice, to compare them against their theoretical worst-case performance ratio of 2. Recall that (\ref{eqn:UB}) can be solved trivially by ordering the jobs $j\in N$ by non-decreasing $p_j+q_j$ values. In our experiments, (\ref{eqn:UB}) was solved in less than 0.01 seconds for every instance. 

 Figure \ref{fig:avg_ub_gap} shows the average relative gap between the solution to (\ref{eqn:UB}) and the solution to (\ref{eqn:recsmsp}) found by solving the MIP, plotted for values of $\Delta$ as a percentage of $n$, for $n\in\{10,20,50,100\}$. A dotted line has been used to indicate the range of $\Delta$ values for which the MIP was not solved to optimality for all instances within the 20 minute time limit for the odd values of $\Delta$ within this range. The average gaps displayed for the odd values of $\Delta$ in this range are therefore estimates of the true average gap. 
 
 For all problem sizes the average gap decreases as $\Delta$ increases, reaching 0 when $\Delta=n$, where (\ref{eqn:recsmsp}) and (\ref{eqn:UB}) become equivalent. For $n=10$, the largest gap between for any single instance was 21.7\%; for $n=20$ the largest single gap was 18.2\%, for $n=50$ it was 23.0\%, and for $n=100$ it was 21.1\%. For each value of $n$, this largest single gap was attained when $\Delta=0$. Clearly, the worst-case gaps we see in practice are considerably smaller than the theoretical worst-case gap of 100\%.  
 
 \begin{figure}[htb]
 	\centering
 	\includegraphics[scale=0.65]{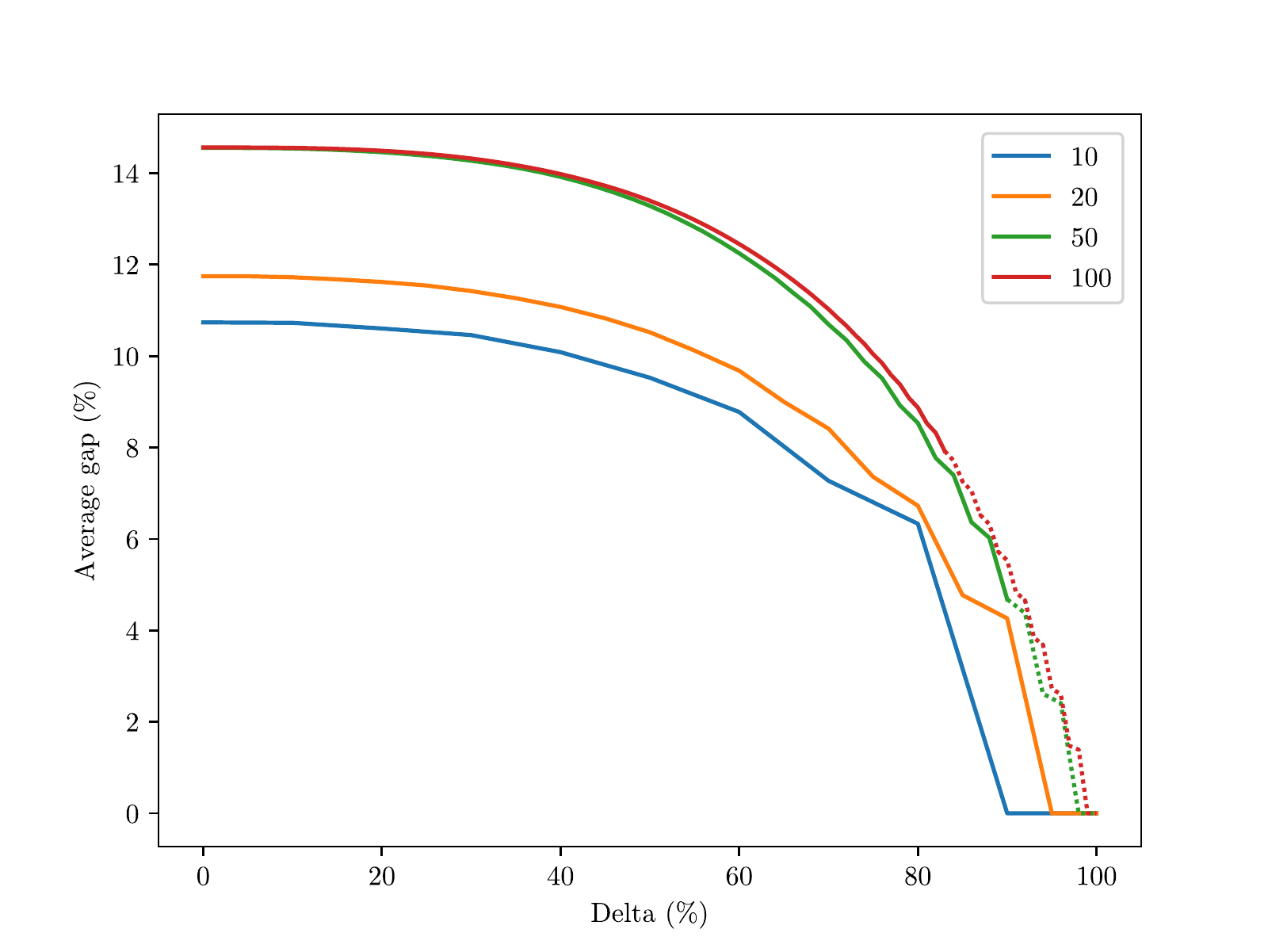}
	\caption{Average gap between solution to (\ref{eqn:UB}) and the solution to (\ref{eqn:recsmsp}) for values of $\Delta$ as a percentage of $n$, for $n\in\{10,20,50,100\}$. Line becomes dotted for the range of $\Delta$ values for which not all instances were solved to optimality by the MIP within its 20 minute time limit for the odd values of $\Delta$ within this range.}
 	\label{fig:avg_ub_gap}
 \end{figure}

 \subsection{Greedy heuristic}\label{section:greedy}
 
 Finally, we examine the performance of greedy heuristic presented in Algorithm \ref{alg:greedy} for solving (\ref{eqn:recsmsp}). Recall that, as a result of Corollary~\ref{cor:2appx}, the objective value of solutions found by the greedy heuristic can be no worse than the solutions to (\ref{eqn:UB}), and as these results show, in practice they are considerably better than this.
 
 We first look at the average run time of the greedy heuristic, shown in Figure \ref{fig:avg_greedy_time} for values of $\Delta$ as a percentage of $n$, for each of the four instance sets. As expected, we see that the average run time grows with $n$ and with $\Delta$. The longest run time of any single instance was 4.72 seconds, which occurred for an instance with $n=100,\,\Delta=95$.
 
 \begin{figure}[htbp]
 	\centering
 	\includegraphics[scale=0.65]{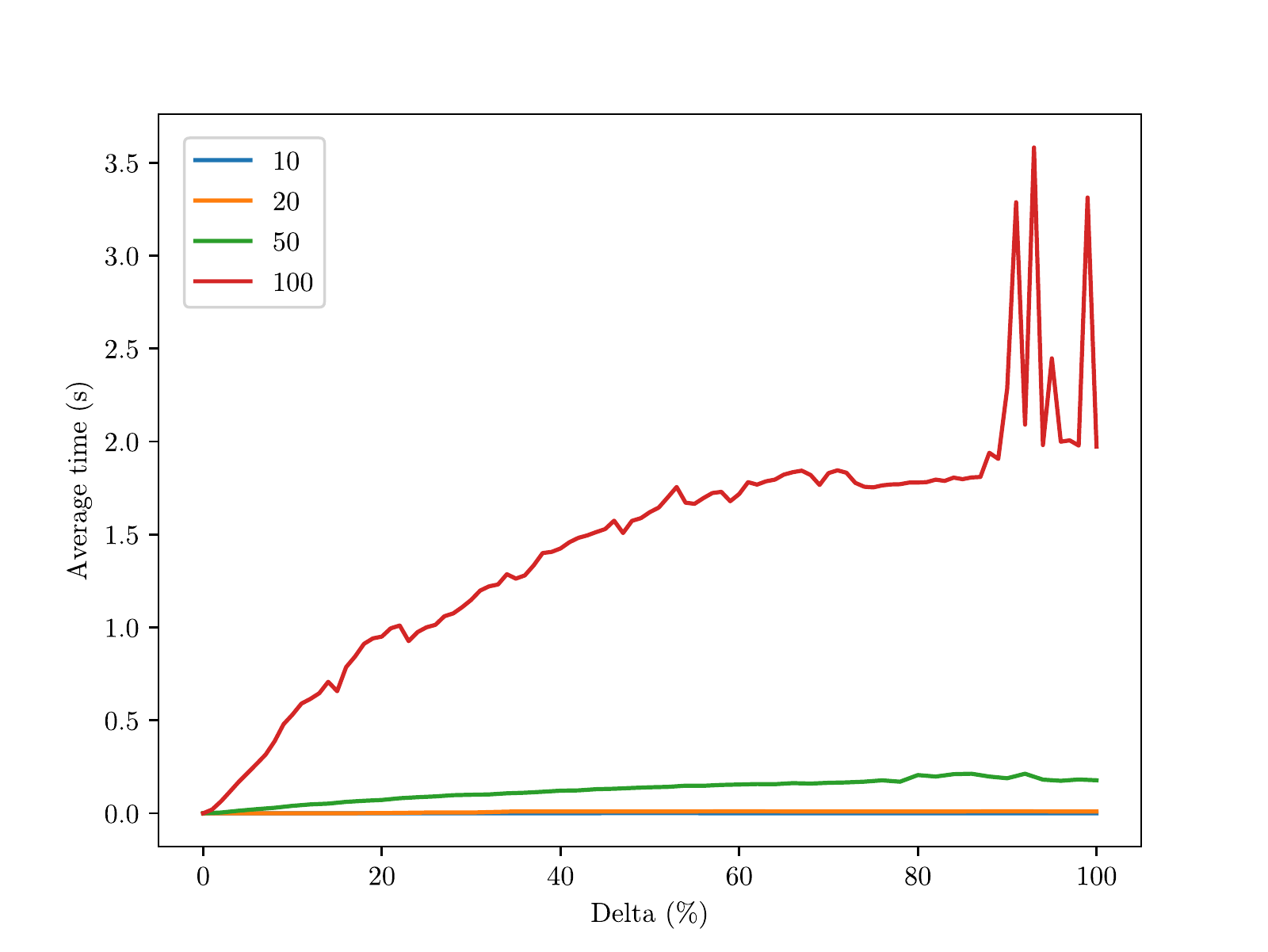}
	\caption{Average greedy heuristic run time for values of $\Delta$ as a percentage of $n$, for $n\in\{10,20,50,100\}$.}
 	\label{fig:avg_greedy_time}
 \end{figure}

 Figure \ref{fig:avg_rel_gap} shows the average relative gap between the solution found by the greedy heuristic and the solution to the MIP for values of $\Delta$ as a percentage of $n$, for $n\in\{10,20,50,100\}$. A dotted line connects the data points at each value of $\Delta$, whilst a solid line connects the data points corresponding to odd values of $\Delta$ only. The figure shows that the accuracy of the greedy heuristic varies depending on whether $\Delta$ is odd or even, demonstrated by the characteristic zig-zagging of the dotted line. 
 
 Even as the average gap between the greedy heuristic and the MIP increases with $\Delta$, this gap is bounded above by the average gap between the solutions to (\ref{eqn:UB}) and the MIP, which sharply decreases with $\Delta$ (see Figure \ref{fig:avg_ub_gap}). The greedy heuristic finds solutions that are an order of magnitude closer to optimal than the solutions to (\ref{eqn:UB}). For $n=10$ the maximum gap between the solution found by the greedy heuristic and the MIP solution for a single instance is 0.55\%, for $n=20$ it is 0.59\%, for $n=50$ it is 0.33\%, and for $n=100$ it is 0.18\%. Hence, solutions found by the greedy heuristic are nearly optimal.
 
%  Hence the greedy heuristic finds solutions that beat even those found by the linear relaxation of the MIP, and is clearly the strongest proposed heuristic approach for solving instances of (\ref{eqn:recsmsp}) that cannot be solved exactly.
 
 \begin{figure}[htbp]
 	\centering
 	\includegraphics[scale=0.65]{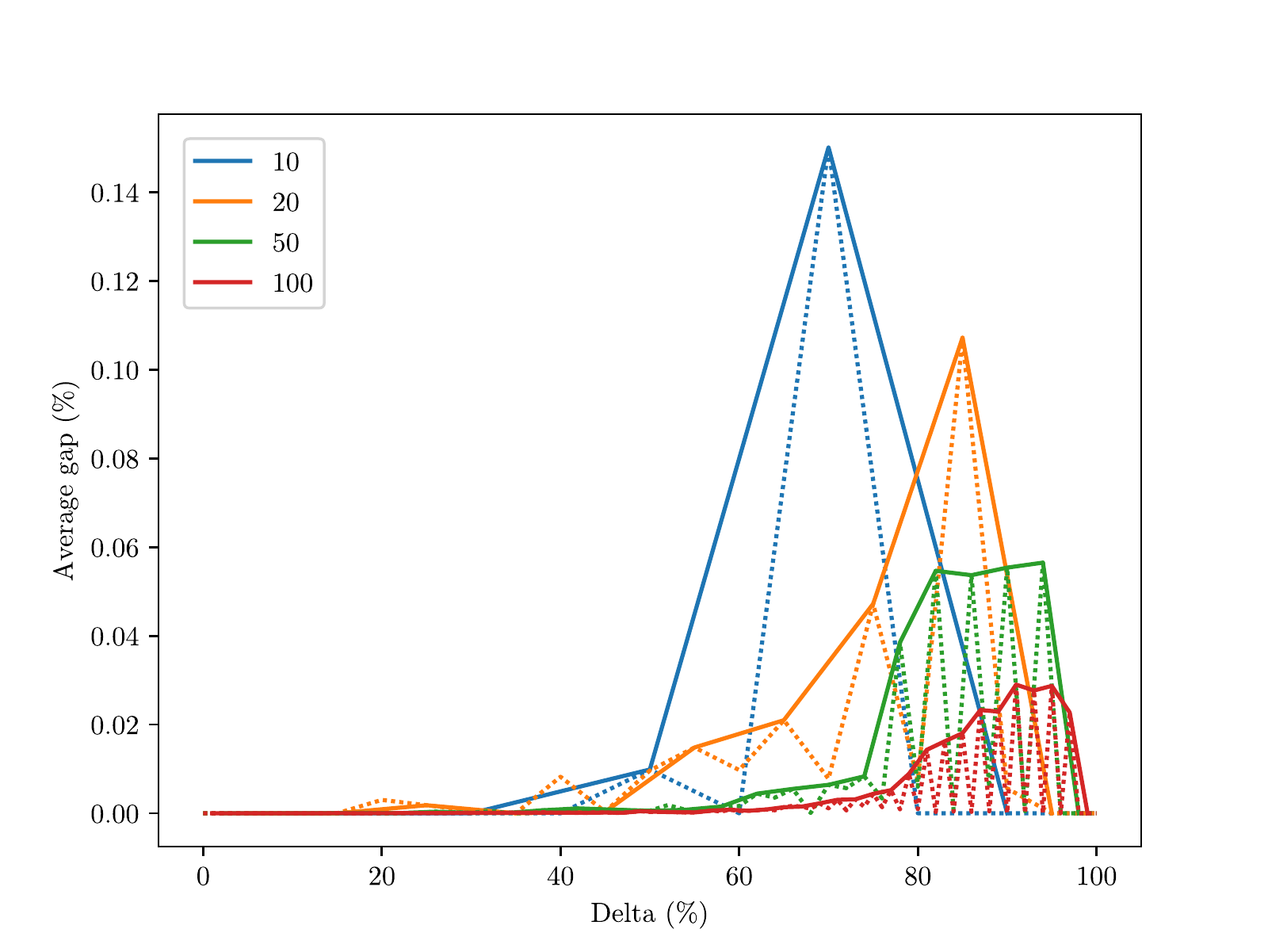}
	\caption{Average gap between greedy and MIP solution for values of $\Delta$ as a percentage of $n$, for $n\in\{10,20,50,100\}$, plotted as a dotted line. Odd-$\Delta$ data points are connected with a solid line.}
 	\label{fig:avg_rel_gap}
 \end{figure}

\section{Conclusions} \label{section:conclusion}

This paper has presented a theoretical and computational investigation of the recoverable robust single machine scheduling problem under interval uncertainty with the objective of minimising the sum of completion times, (\ref{eqn:recsmsp}). Firstly, a number of positive complexity results have been produced for key subproblems. Specifically, we present a polynomial time sorting algorithm in Algorithm \ref{alg:eval} which can optimally solve the case where a subset of jobs that must share a first and second-stage assignment is given. Using this result, we show that the special cases of (\ref{eqn:recsmsp}) with a constant value of $\Delta$ and with a constant number of possible processing times are solvable in polynomial time. We then introduce the problem (\ref{eqn:UB}) in which the first and second-stage assignments are forced to be identical, and prove that the solution to this problem provides a 2-approximation of (\ref{eqn:recsmsp}). This guarantee extends to the solutions found by a simple and fast greedy heuristic.
A set of extensive computational experiments investigate the limitations of an exact mixed-integer programming approach to solving the problem, show that the solutions of (\ref{eqn:UB}) significantly outperform their worst-case guarantee in practice, and demonstrate the strength of the proposed greedy heuristic.

Note that the computational complexity of (\ref{eqn:recsmsp}) has yet to be characterised and remains unknown. In addition to investigating this, a promising direction for future research on this problem would be to consider the impact of alternative uncertainty sets on its complexity. Finally, another possibility for future research would be to look at this recoverable robust model for single machine scheduling problems with objective criteria different to the sum of completion times that we consider here.

\section*{Acknowledgements}

The authors are grateful for the support of the EPSRC-funded (EP/L015692/1) STOR-i Centre for Doctoral Training.

\bibliography{project4}

\begin{thebibliography}{}

\bibitem[Aloulou and Della~Croce, 2008]{aloulou2008complexity}
Aloulou, M.~A. and Della~Croce, F. (2008).
\newblock Complexity of single machine scheduling problems under scenario-based
  uncertainty.
\newblock {\em Operations Research Letters}, 36(3):338--342.

\bibitem[Bold and Goerigk, 2020]{bold2020recoverable}
Bold, M. and Goerigk, M. (2020).
\newblock Recoverable robust single machine scheduling with budgeted
  uncertainty.
\newblock {\em arXiv preprint arXiv:2011.06284}.

\bibitem[Bougeret et~al., 2019]{bougeret2019robust}
Bougeret, M., Pessoa, A.~A., and Poss, M. (2019).
\newblock Robust scheduling with budgeted uncertainty.
\newblock {\em Discrete Applied Mathematics}, 261:93--107.

\bibitem[B{\"u}sing, 2012]{busing2012recoverable}
B{\"u}sing, C. (2012).
\newblock Recoverable robust shortest path problems.
\newblock {\em Networks}, 59(1):181--189.

\bibitem[Daniels and Kouvelis, 1995]{daniels1995robust}
Daniels, R.~L. and Kouvelis, P. (1995).
\newblock Robust scheduling to hedge against processing time uncertainty in
  single-stage production.
\newblock {\em Management Science}, 41(2):363--376.

\bibitem[De{\v\i}neko et~al., 2006]{dei2006robust}
De{\v\i}neko, V.~G., Woeginger, G.~J., et~al. (2006).
\newblock On the robust assignment problem under a fixed number of cost
  scenarios.
\newblock {\em Operations Research Letters}, 34(2):175--179.

\bibitem[Fischer et~al., 2020]{fischer2020investigation}
Fischer, D., Hartmann, T.~A., Lendl, S., and Woeginger, G.~J. (2020).
\newblock An investigation of the recoverable robust assignment problem.
\newblock {\em arXiv preprint arXiv:2010.11456}.

\bibitem[Goerigk et~al., 2021]{goerigk2021recoverable}
Goerigk, M., Lendl, S., and Wulf, L. (2021).
\newblock On the recoverable traveling salesman problem.
\newblock {\em arXiv preprint arXiv:2111.09691}.

\bibitem[Graham et~al., 1979]{graham1979optimization}
Graham, R.~L., Lawler, E.~L., Lenstra, J.~K., and Kan, A.~R. (1979).
\newblock Optimization and approximation in deterministic sequencing and
  scheduling: a survey.
\newblock In {\em Annals of Discrete Mathematics}, volume~5, pages 287--326.
  Elsevier.

\bibitem[Hradovich et~al., 2017a]{hradovich2017recoverable}
Hradovich, M., Kasperski, A., and Zieli{\'n}ski, P. (2017a).
\newblock Recoverable robust spanning tree problem under interval uncertainty
  representations.
\newblock {\em Journal of Combinatorial Optimization}, 34(2):554--573.

\bibitem[Hradovich et~al., 2017b]{hradovich2017recoverableb}
Hradovich, M., Kasperski, A., and Zieli{\'n}ski, P. (2017b).
\newblock The recoverable robust spanning tree problem with interval costs is
  polynomially solvable.
\newblock {\em Optimization Letters}, 11(1):17--30.

\bibitem[Kasperski and Zieli{\'n}ski, 2008]{kasperski20082}
Kasperski, A. and Zieli{\'n}ski, P. (2008).
\newblock A 2-approximation algorithm for interval data minmax regret
  sequencing problems with the total flow time criterion.
\newblock {\em Operations Research Letters}, 36(3):343--344.

\bibitem[Kasperski and Zielinski, 2014]{kasperski2014minmax}
Kasperski, A. and Zielinski, P. (2014).
\newblock Minmax (regret) scheduling problems.
\newblock {\em Sequencing and scheduling with inaccurate data}, pages 159--210.

\bibitem[Kasperski and Zieli{\'n}ski, 2016a]{kasperski2016robust}
Kasperski, A. and Zieli{\'n}ski, P. (2016a).
\newblock Robust discrete optimization under discrete and interval uncertainty:
  A survey.
\newblock In {\em Robustness analysis in decision aiding, optimization, and
  analytics}, pages 113--143. Springer.

\bibitem[Kasperski and Zieli{\'n}ski, 2016b]{kasperski2016single}
Kasperski, A. and Zieli{\'n}ski, P. (2016b).
\newblock Single machine scheduling problems with uncertain parameters and the
  {OWA} criterion.
\newblock {\em Journal of Scheduling}, 19(2):177--190.

\bibitem[Kasperski and Zieli{\'n}ski, 2017]{kasperski2017robust}
Kasperski, A. and Zieli{\'n}ski, P. (2017).
\newblock Robust recoverable and two-stage selection problems.
\newblock {\em Discrete Applied Mathematics}, 233:52--64.

\bibitem[Kasperski and Zieli{\'n}ski, 2019]{kasperski2019risk}
Kasperski, A. and Zieli{\'n}ski, P. (2019).
\newblock Risk-averse single machine scheduling: complexity and approximation.
\newblock {\em Journal of Scheduling}, 22(5):567--580.

\bibitem[Kouvelis and Yu, 1997]{kouvelis1997robust}
Kouvelis, P. and Yu, G. (1997).
\newblock {\em Robust discrete optimization and its applications}.
\newblock Kluwer Academic Publishers Dordrecht, Netherlands.

\bibitem[Lebedev and Averbakh, 2006]{lebedev2006complexity}
Lebedev, V. and Averbakh, I. (2006).
\newblock Complexity of minimizing the total flow time with interval data and
  minmax regret criterion.
\newblock {\em Discrete Applied Mathematics}, 154(15):2167--2177.

\bibitem[Liebchen et~al., 2009]{liebchen2009concept}
Liebchen, C., L{\"u}bbecke, M., M{\"o}hring, R., and Stiller, S. (2009).
\newblock The concept of recoverable robustness, linear programming recovery,
  and railway applications.
\newblock In {\em Robust and online large-scale optimization}, pages 1--27.
  Springer.

\bibitem[Lu et~al., 2014]{lu2014robust}
Lu, C.-C., Ying, K.-C., and Lin, S.-W. (2014).
\newblock Robust single machine scheduling for minimizing total flow time in
  the presence of uncertain processing times.
\newblock {\em Computers \& Industrial Engineering}, 74:102--110.

\bibitem[Mastrolilli et~al., 2013]{mastrolilli2013single}
Mastrolilli, M., Mutsanas, N., and Svensson, O. (2013).
\newblock Single machine scheduling with scenarios.
\newblock {\em Theoretical Computer Science}, 477:57--66.

\bibitem[Montemanni, 2007]{montemanni2007mixed}
Montemanni, R. (2007).
\newblock A mixed integer programming formulation for the total flow time
  single machine robust scheduling problem with interval data.
\newblock {\em Journal of Mathematical Modelling and Algorithms},
  6(2):287--296.

\bibitem[Pereira and Averbakh, 2011]{pereira2011exact}
Pereira, J. and Averbakh, I. (2011).
\newblock Exact and heuristic algorithms for the interval data robust
  assignment problem.
\newblock {\em Computers \& Operations Research}, 38(8):1153--1163.

\bibitem[Yang and Yu, 2002]{yang2002robust}
Yang, J. and Yu, G. (2002).
\newblock On the robust single machine scheduling problem.
\newblock {\em Journal of Combinatorial Optimization}, 6(1):17--33.

\bibitem[Zhao et~al., 2010]{zhao2010family}
Zhao, H., Zhao, M., et~al. (2010).
\newblock A family of inequalities valid for the robust single machine
  scheduling polyhedron.
\newblock {\em Computers \& Operations Research}, 37(9):1610--1614.

\end{thebibliography}

\end{document}